\documentclass[a4paper]{article} \def\nodraft {} \pdfoutput=1

\usepackage[a4paper]{geometry} \usepackage{amsmath} \usepackage{amssymb} \usepackage{subcaption}

\usepackage{amsthm}

\usepackage{tikz} \usetikzlibrary{topaths,calc} \usetikzlibrary{positioning,shapes.geometric} \usetikzlibrary{decorations.markings} \usetikzlibrary{decorations.pathmorphing} \usetikzlibrary{patterns}

\newtheorem{thm}{Theorem} \newtheorem{lmm}{Lemma} \newtheorem{crll}{Corollary} 

\theoremstyle{definition} \newtheorem{defn}{Definition} \newtheorem{prop}{Proposition} \newtheorem{expl}{Example}

\newcommand{\Hext}{H_\text{\normalfont{ext}}} \newcommand{\Hint}{H_\text{\normalfont{int}}}

 \DeclareMathOperator{\res}{res} \DeclareMathOperator{\id}{id} \DeclareMathOperator{\im}{im} \DeclareMathOperator{\Aut}{Aut} 

\newcommand{\meet}{\wedge} \newcommand{\join}{\vee}

\newcommand{\hopffg}{\mathcal{H}^\text{\normalfont{fg}}} \newcommand{\hopffgs}{\widetilde{\mathcal{H}}^\text{\normalfont{fg}}} \newcommand{\hopfpos}{\mathcal{H}^\text{\normalfont{P}}} \newcommand{\hopflat}{\mathcal{H}^\text{\normalfont{L}}} \newcommand{\hopflats}{\widetilde{\mathcal{H}}^\text{\normalfont{L}}} \newcommand{\unit}{\text{\normalfont{u}}} \newcommand{\counit}{\epsilon}

\newcommand{\subdiags}{\mathcal{P}} \newcommand{\sdsubdiags}{\mathcal{P}^\text{s.d.}} \newcommand{\sdsubdiagsn}{\widetilde{\mathcal{P}}^\text{s.d.}} 

\newcommand{\fourvtx}{ \ifdefined\nodraft     {     \begin{tikzpicture}[x=1ex,y=1ex,baseline={([yshift=-.5ex]current bounding box.center)}] \coordinate (v) ; \def \n {4}; \def \rad {1}; \foreach \s in {1,...,5} { \def \angle {360/\n*(\s - 1)+45}; \coordinate (u) at ([shift=({\angle}:\rad)]v); \draw (v) -- (u); } \filldraw (v) circle (1pt); \end{tikzpicture}     } \else X
\fi }

\newcommand{\fourvtxgluon}{ \ifdefined\nodraft     {     \begin{tikzpicture}[x=1.5ex,y=1.5ex,baseline={([yshift=-.5ex]current bounding box.center)}] \coordinate (v) ; \def \n {4}; \def \rad {1}; \foreach \s in {1,...,5} { \def \angle {360/\n*(\s - 1)+45}; \coordinate (u) at ([shift=({\angle}:\rad)]v); \draw[gluon] (v) -- (u); } \filldraw (v) circle (1pt); \end{tikzpicture}     } \else X
\fi }

\newcommand{\twothreevtxgluon}{ \ifdefined\nodraft     {     \begin{tikzpicture}[x=1.5ex,y=1.5ex,baseline={([yshift=-.5ex]current bounding box.center)}] \coordinate (i1); \coordinate[below=1 of i1] (i2); \coordinate[right=1 of i1] (v1); \coordinate[right=1 of i2] (v2); \coordinate[right=1 of v1] (o1); \coordinate[right=1 of v2] (o2); \draw[gluon] (i1) -- (v1); \draw[gluon] (i2) -- (v2); \draw[gluon] (v1) -- (o1); \draw[gluon] (v1) -- (v2); \draw[gluon] (v2) -- (o2); \filldraw (v1) circle (1pt); \filldraw (v2) circle (1pt); \end{tikzpicture}     } \else I
\fi }

\newcommand{\simpleprop}{ \ifdefined\nodraft     {     \begin{tikzpicture}[x=1ex,y=1ex,baseline={([yshift=-.5ex]current bounding box.center)}] \coordinate (v) ; \coordinate [right=1 of v] (u); \draw (v) -- (u); \end{tikzpicture}     } \else -
\fi }

\newcommand{\oneloopprop}{ \ifdefined\nodraft     {     \begin{tikzpicture}[x=2ex,y=2ex,baseline={([yshift=-.5ex]current bounding box.center)}] \coordinate (i0) ; \coordinate [right=.5 of i0] (v0); \coordinate [right=1 of v0] (v1); \coordinate [right=.5 of v1] (o0); \coordinate [right=.5 of v0] (vm); \draw (vm) circle(.5); \draw (i0) -- (v0); \draw (o0) -- (v1); \filldraw (v0) circle(1pt); \filldraw (v1) circle(1pt); \end{tikzpicture}     } \else -o- \fi }

\pgfdeclaredecoration{complete sines}{initial} {
    \state{initial}[         width=+0pt,         next state=sine,         persistent precomputation={\pgfmathsetmacro\matchinglength{             \pgfdecoratedinputsegmentlength / int(\pgfdecoratedinputsegmentlength/\pgfdecorationsegmentlength)}             \setlength{\pgfdecorationsegmentlength}{\matchinglength pt}         }] {}     \state{sine}[width=\pgfdecorationsegmentlength]{         \pgfpathsine{\pgfpoint{0.25\pgfdecorationsegmentlength}{0.5\pgfdecorationsegmentamplitude}}         \pgfpathcosine{\pgfpoint{0.25\pgfdecorationsegmentlength}{-0.5\pgfdecorationsegmentamplitude}}         \pgfpathsine{\pgfpoint{0.25\pgfdecorationsegmentlength}{-0.5\pgfdecorationsegmentamplitude}}         \pgfpathcosine{\pgfpoint{0.25\pgfdecorationsegmentlength}{0.5\pgfdecorationsegmentamplitude}} }
    \state{final}{} }

\tikzset{     photon/.style={         decoration={complete sines, amplitude=0.15cm, segment length=0.2cm},         decorate         },     fermion/.style={         decoration={             markings,             mark=at position 0.5 with {\node[transform shape, xshift=-0.5mm, fill=black, inner sep=1pt, draw, isosceles triangle]{};}         },         postaction=decorate     },     gluon/.style={         decoration={coil, aspect=0.75, mirror, amplitude=.4mm, segment length=.8mm},         decorate     },      meson/.style={         dashed     },      left/.style={         bend left=90,         looseness=1.75     },      leftsoft/.style={         bend left=45,         looseness=1.25     } }

\tikzset{every loop/.style={looseness=12,min distance=1cm}}

\title{Algebraic lattices in QFT renormalization} \author{Michael Borinsky\footnote{borinsky@physik.hu-berlin.de}\\ Institute of Physics,  Humboldt University \\ Newton Str. 15,  D-12489 Berlin, Germany }
\date{}

\begin{document} \maketitle

\bibliographystyle{plain}

\begin{abstract} The structure of overlapping subdivergences, which appear in the perturbative expansions of quantum field theory, is analyzed using algebraic lattice theory. It is shown that for specific QFTs the sets of subdivergences of Feynman diagrams form algebraic lattices. This class of QFTs includes the Standard model. In kinematic renormalization schemes, in which \textit{tadpole} diagrams vanish, these lattices are semimodular. This implies that the Hopf algebra of Feynman diagrams is graded by the coradical degree or equivalently that every maximal forest has the same length in the scope of BPHZ renormalization. As an application of this framework a formula for the counter terms in zero-dimensional QFT is given together with some examples of the enumeration of primitive or skeleton diagrams.


 \end{abstract}

\section{Introduction}

Calculations of observable quantities in quantum field theory rely almost always on perturbation theory. The integrals in the perturbation expansions can be depicted as Feynman diagrams. Usually these integrals require \textit{renormalization} to give meaningful results. The renormalization procedure can be performed using the Hopf algebra of Feynman graphs \cite{connes2001renormalization}, which organizes the classic BPHZ procedure into an algebraic framework. 

The motivation for this paper was to obtain insights on the \textit{coradical filtration} of this Hopf algebra and thereby on the structure of the \textit{subdivergences} of the Feynman diagrams in the perturbation expansion. 

The perturbative expansions in QFT are divergent series themselves as first pointed out by \cite{dyson1952divergence}. Except for objects called \textit{renormalons} \cite{lautrup1977high}, this divergence is believed to be dominated by the growth of the number of Feynman diagrams. The coradical filtration describes the hierarchy in which diagrams become important in the large-order regime. Dyson-Schwinger equations exploit this hierarchy to give non-perturbative results \cite{kruger2015filtrations}.  This work also aims to extend the effectiveness of these methods.

Notation and preliminaries on Feynman diagrams and Hopf algebras are covered in section \ref{sec:pre} in a form suitable for a combinatorial analysis. Feynman diagrams are defined as a special case of hyper graphs as in \cite{Yeats2008}. This definition was used to clarify the role of external legs and isomorphisms of diagrams. Based on this, Kreimer's Hopf algebra of Feynman diagrams is defined.

The starting point in section \ref{sec:posetsandlattices} for the analysis of algebraic lattices in renormalization is the basic fact that \textit{subdivergences} of Feynman diagrams form a partially ordered set or poset ordered by inclusion. 

The idea to search for more properties of the subdivergence posets was inspired by the work of Berghoff \cite{berghoff2014wonderful}, who studied the posets of subdivergences in the context of Epstein-Glaser renormalization and proved that the subdivergences of diagrams with only logarithmic divergent subdivergences form distributive lattices. Distributive lattices have already been used in \cite[Part III]{figueroa2005combinatorial} to describe subdivergences of Feynman diagrams.

It is straightforward to carry over the Hopf algebra structure to the posets and lattices by using the the incidence Hopf algebra on posets \cite{Schmitt1994} as shown in section \ref{sec:posetsandlattices}.

\paragraph{Statement of results} Only in distinguished renormalizable quantum field theories a \textit{join} and a \textit{meet} can be generally defined on the posets of subdivergences of Feynman diagrams, promoting the posets to \textit{algebraic lattices}. These distinguished renormalizable QFTs will be called \textit{join-meet-renormalizable}. It is shown that a broad class of QFTs including the standard model falls into this category. $\phi^6$-theory in $3$-dimensions is examined as an example of a QFT, which is renormalizable, but not join-meet-renormalizable. 

A further analysis in section \ref{sec:properties} demonstrates that in QFTs with only three-or-less-valent vertices, which are thereby join-meet-renormalizable, these lattices are \textit{semimodular}. This implies that the Hopf algebra is bigraded by the loop number of the Feynman diagram and its \textit{coradical degree}. In the language of BPHZ this means that every \textit{complete forest} has the same length. Generally, this structure cannot be found in join-meet-renormalizable theories with also four-valent vertices as QCD or $\phi^4$. An explicit counter example of a non-graded and non-semimodular lattice, which appears in $\phi^4$ and Yang-Mills theories, is given. The semimodularity of the subdivergence lattices can be resurrected in these cases by dividing out \textit{tadpole} (also snail or seagull) diagrams. This quotient can always be formed in kinematic renormalization schemes. 

This whole framework is used in section \ref{sec:applications} to illuminate some results of zero-dimensional QFTs \cite{cvitanovic1978number,argyres2001zero} from the perspective of the Hopf algebra of decorated lattices. A closed formula for the counter term calculation is given using the \textit{Moebius function} on the lattices. It is used to enumerate primitive diagrams in a variety of cases.

\section{Preliminaries} \label{sec:pre} \subsection{Combinatorial quantum field theory}     In what follows a quantum field theory (QFT) will be characterized by its field content, its interactions, associated `weights' for these interactions and a given dimension of spacetime $D$. Let $\Phi$ denote the set of fields, $\mathcal{R}_v$ the set of allowed interactions, represented as monomials in the fields and $\mathcal{R}_e \subset \mathcal{R}_v$ the set of propagators, a set of distinguished interactions between two fields only. $\mathcal{R}_e$ consists of monomials of degree two and $\mathcal{R}_v$ of monomials of degree two or higher in the fields $\Phi$. Additionally, a map $\omega: \mathcal{R}_e \cup \mathcal{R}_v \rightarrow \mathbb{Z}$ is given associating a weight to each interaction. 

This approach for the definition of Feynman diagrams based on half-edges is well-known. See for instance \cite[Sec. 2.3]{Yeats2008} or \cite[Sec. 2.1]{gurau2014renormalization}.

The requirement $\mathcal{R}_e \subset \mathcal{R}_v$ ensures that there is a two-valent vertex for every allowed edge-type. This is not necessary for the definition of the Hopf algebra of Feynman diagrams, but it results in a simpler formula for contractions among other simplifications. Of course, this does not introduce a restriction to the underlying QFT: A propagator is always associated to the formal inverse of the corresponding two-valent vertex and a two-valent vertex always comes with an additional propagator in a diagram. The two valent vertex of the same type as the propagator can be canceled with the additional propagator.

In physical terms, the interactions correspond to summands in the Lagrangian of the QFT and the weights are the number of derivatives in the respective summand. Having clarified the important properties of a QFT for a combinatorial treatment, we can proceed to the definition of the central object of perturbative QFTs: \subsection{Feynman diagrams as hypergraphs} \label{sec:feynmandiagrams}

\begin{defn}[Feynman diagram] \label{def:feynmandiagram}   A \textit{Feynman diagram} $\Gamma$ is a tuple $\left( H,E,V,\eta \right)$.    Consisting of  \begin{enumerate}     \item a set $H$ of half-edges, \label{itm:firsthg}     \item a coloring of the half-edges by fields in $\Phi$:      \begin{align} \eta: H \rightarrow \Phi. \end{align}     This coloring also induces an additional map, the \textit{residue},     \begin{align} \res&: 2^{H} \rightarrow (\mathbb{N}_0)^{\Phi}, & a &\mapsto \prod \limits_{h\in a} \eta(h), \end{align}     of subsets of half-edges, or arbitrary adjacency relations, to monomials in the fields,     \item a set $E$ of edges, adjacency relations between two half-edges with a residue in $\mathcal{R}_e$,     \begin{align} E \subset \left\{ e \subset H : \res(e) \in \mathcal{R}_e \right\} \text{ and} \end{align}     \item a set $V$ of vertices or corollas - adjacency relations between any other number of half-edges with a residue in $\mathcal{R}_v$:      \begin{align} V \subset \left\{ v \subset H : \res(v) \in \mathcal{R}_v \right\}. \end{align}     \label{itm:lasthg}

With the conditions: 

    \item All the edges and all the vertices are each pairwise disjoint sets: $\forall v_1, v_2 \in V,v_1 \neq v_2: v_1 \cap v_2 = \emptyset$ as well as $\forall e_1, e_2 \in E,e_1 \neq e_2: e_1 \cap e_2 = \emptyset$. \label{cond:disjedges}      \item Every half-edges is in at least one vertex: $\bigcup \limits _{v \in V} v = H$. \label{cond:disjvtcs}  \end{enumerate} \end{defn}

Conditions \ref{itm:firsthg}-\ref{itm:lasthg} form the definition of a \textit{colored hyper graph}. The other conditions \ref{cond:disjedges} and \ref{cond:disjvtcs} almost restrict these hyper graphs to \textit{multigraphs} with the exception that half-edges that are not in any edge are still allowed. These half-edges will play the role of the external legs of the Feynman diagram. Feynman diagrams will also be called diagrams or graphs in this article.  If the reference to the diagram is ambiguous, the sets in the tuple $\Gamma = \left( H,E,V,\eta \right)$ will be denoted as $H(\Gamma)$, $E(\Gamma)$ and $V(\Gamma)$. To clarify the above definition an example is given, in which different depictions of Feynman diagrams are discussed. \begin{figure} \ifdefined\nodraft   \subcaptionbox{Typical graph representation of a Feynman diagram.\label{fig:traditionalfg}}   [.45\linewidth]{     \begin{tikzpicture} \coordinate (v1); \coordinate[right=.5 of v1] (v2); \coordinate[right=of v2] (v3); \coordinate[right=.5 of v3] (v4); \draw[fermion] (v1) -- (v2); \draw[meson] (v2) to[left] (v3); \draw[fermion] (v2) -- (v3); \draw[fermion] (v3) -- (v4); \end{tikzpicture}   }   \subcaptionbox{Hyper graph representation of a Feynman diagram.\label{fig:hyperfeynmangraph}}   [.45\linewidth]{     \begin{tikzpicture}[scale=.8] \node (p1) at (0,0.75) {$\bar \psi$}; \node (f1) at (1,1.5) {$\varphi$}; \node (pb1) at (1,0) {$\psi$}; \node (p2) at (5,0) {$\bar \psi$}; \node (f2) at (5,1.5) {$\varphi$}; \node (pb2) at (6,0.75) {$\psi$}; \node (c1) at ($1/3*(p1)+1/3*(pb1)+1/3*(f1)$) {}; \node (c2) at ($1/3*(p2)+1/3*(pb2)+1/3*(f2)$) {}; \node at ($(c1) + (0.35,0)$) {$\bar \psi \varphi \psi$}; \node at ($(c2) - (0.35,0)$) {$\bar \psi \varphi \psi$}; \node (e1) at ($1/2*(pb1) + 1/2*(p2)$) {$\bar{\psi} \psi$}; \node (e2) at ($1/2*(f1) + 1/2*(f2)$) {$\varphi \varphi$}; \draw (p1) circle (0.25); \draw (p2) circle (0.25); \draw (pb1) circle (0.25); \draw (pb2) circle (0.25); \draw (f1) circle (0.25); \draw (f2) circle (0.25); \draw (c1) circle (1.5); \draw (c2) circle (1.5); \draw (e1) ellipse [x radius=2.5, y radius=0.5]; \draw (e2) ellipse [x radius=2.5, y radius=0.5]; \end{tikzpicture}   }   \subcaptionbox{Bipartite graph representation of the Feynman diagram.\label{fig:bipartiefeynmangraph}}   [\linewidth]{     \begin{tikzpicture}[main node/.style={circle,draw}] \node[main node] (h1) {$\bar \psi$}; \node[main node] (h2) [right=of h1] {$\varphi$}; \node[main node] (h3) [right=of h2] {$\psi$}; \node[main node] (h4) [right=of h3] {$\varphi$}; \node[main node] (h5) [right=of h4] {$\bar \psi$}; \node[main node] (h6) [right=of h5] {$\psi$}; \node[main node] (a1) [below=of h2] {$\bar \psi \varphi \psi$}; \node[main node] (a2) [below=of h3] {$\varphi \varphi$}; \node[main node] (a3) [below=of h4] {$\bar{\psi} \psi$}; \node[main node] (a4) [below=of h5] {$\bar \psi \varphi \psi$}; \draw (h1) -- (a1); \draw (h2) -- (a1); \draw (h3) -- (a1); \draw (h4) -- (a4); \draw (h5) -- (a4); \draw (h6) -- (a4); \draw[draw=white,double=black,very thick] (h2) -- (a2); \draw[draw=white,double=black,very thick] (h4) -- (a2); \draw[draw=white,double=black,very thick] (h3) -- (a3); \draw[draw=white,double=black,very thick] (h5) -- (a3); \end{tikzpicture}   } \else

MISSING IN DRAFT MODE

\fi   \caption{Equivalent diagrammatic representations of Feynman graphs}\label{fig:diagramaticfeynmangraphs} \end{figure}
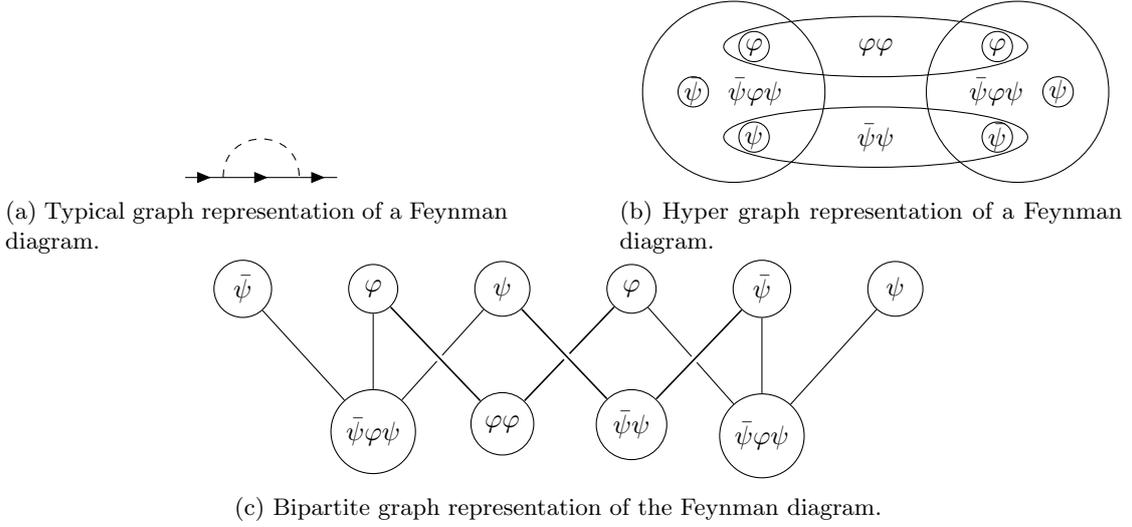 \begin{expl}[Yukawa theory] Let $\Phi = \left\{ \bar \psi, \psi, \phi \right\}$, $\mathcal R_v = \left\{ \bar \psi \psi, \phi^2, \bar \psi \psi \phi \right\}$ and $\mathcal R_e =\left\{ \bar \psi \psi, \phi^2 \right\}$ . Fig. \ref{fig:diagramaticfeynmangraphs} shows different graphical representations for a simple Feynman diagram in this theory. 

The usual Feynman diagram representation is given in fig. \ref{fig:traditionalfg}. The adjacency relations $E$ are represented as edges and the adjacency relations $V$ as vertices. The half-edges are omitted. 

Fig. \ref{fig:hyperfeynmangraph} shows a hypergraph representation of the diagram. Its half-edges are drawn as little circles. They are colored by the corresponding field.  The adjacency relations are shown as big ellipses, enclosing the adjacent half-edges. The adjacency relations, $a\in E\cup V$ can be colored by the different allowed residues, $\res(a)$ in $\mathcal{R}_e$ and $\mathcal{R}_v$.

In fig. \ref{fig:bipartiefeynmangraph}, the same diagram is depicted as a bipartite graph. Half-edges are identified with the first class of vertices of the graph and adjacency relations with the second. Half-edges and an adjacency relation are connected if the half-edge appears in the respective adjacency relation. This representation is useful for computational treatments of Feynman diagrams, because colored bipartite graphs are common and well studied in computational graph theory. \end{expl}

\paragraph{Isomorphisms of diagrams}   An \textit{isomorphism} $j: \Gamma \rightarrow \Gamma'$ between one Feynman diagram $\Gamma=(H,V,E,\eta)$ and another $\Gamma'=(H',V',E',\eta')$ is an bijection   $j: H \rightarrow H'$,  which preserves the adjacency structure    $j\left( E \cup V \right) = E' \cup V'$ and the coloring   $\eta = \eta' \circ j$,  where $j$ is extended canonically to $E \cup V \subset 2^H$. If there is an isomorphism from one diagram $\Gamma$ to another diagram $\Gamma'$, the two diagrams are called isomorphic.  An automorphism is an isomorphism of the diagram onto itself.  The group of automorphisms of $\Gamma$ is denoted as $\Aut \Gamma$. The inverse of the cardinality of the automorphism group, $\frac{1}{|\Aut \Gamma|}$, is the \textit{symmetry factor} of the diagram $\Gamma$.

\paragraph{External legs} The half-edges which are not included in an edge are called \textit{external legs} or \textit{legs}, $\Hext := H \setminus \bigcup \limits_{e \in E} e$. $\Hint=\bigcup \limits_{e \in E} e$ is the set of half-edges which are not external. Feynman diagrams can be classified by the colors of the half-edges in $\Hext$. This classification is called the \textit{residue} of the diagram.

\paragraph{Residues of a diagrams} Using the external half-edges, $\Hext$, a monomial $\prod \limits_{h \in H_\text{ext}} \eta\left( h \right)$ in the fields can be associated to a diagram.  This map will also be called $\res$: \begin{align} \res: \left( H,E,V,\eta \right) \rightarrow (\mathbb{N}_0)^\Phi, \Gamma \mapsto \prod \limits_{h \in H_\text{ext}(\Gamma)} \eta\left( h \right). \end{align} The choice of the same name for this map and the coloring of the adjacency relations will be justified by a certain compatibility of $\res$ for adjacency relations and diagrams in the scope of \textit{renormalizable} QFTs (see definition \ref{def:renormalizable}). Furthermore, this should underline the operadic structure of Feynman diagrams.

\paragraph{(Non)-fixed external legs} In the above definition the external legs of a diagram do not have a specific labeling as usual. The legs are not `fixed' and their permutations can result in additional automorphisms.  For instance, the diagrams $ { \def \scale {2ex} \begin{tikzpicture}[x=\scale,y=\scale,baseline={([yshift=-.5ex]current bounding box.center)}] \begin{scope}[node distance=1] \coordinate (v0); \coordinate [right=.5 of v0] (vm); \coordinate [right=.5 of vm] (v1); \coordinate [above left=.5 of v0] (i0); \coordinate [below left=.5 of v0] (i1); \coordinate [above right=.5 of v1] (o0); \coordinate [below right=.5 of v1] (o1); \draw (vm) circle(.5); \draw (i0) -- (v0); \draw (i1) -- (v0); \draw (o0) -- (v1); \draw (o1) -- (v1); \filldraw (v0) circle(1pt); \filldraw (v1) circle(1pt); \end{scope} \end{tikzpicture} } \simeq { \def \scale {2ex} \begin{tikzpicture}[x=\scale,y=\scale,baseline={([yshift=-.5ex]current bounding box.center)}] \begin{scope}[node distance=1] \coordinate (v0); \coordinate [below=1 of v0] (v1); \coordinate [above left=.5 of v0] (i0); \coordinate [above right=.5 of v0] (i1); \coordinate [below left=.5 of v1] (o0); \coordinate [below right=.5 of v1] (o1); \coordinate [above=.5 of v1] (vm); \draw (vm) circle(.5); \draw (i0) -- (v0); \draw (i1) -- (v0); \draw (o0) -- (v1); \draw (o1) -- (v1); \filldraw (v0) circle(1pt); \filldraw (v1) circle(1pt); \end{scope} \end{tikzpicture} } \simeq { \def \scale {2ex} \begin{tikzpicture}[x=\scale,y=\scale,baseline={([yshift=-.5ex]current bounding box.center)}] \begin{scope}[node distance=1] \coordinate (v0); \coordinate [below=1 of v0] (v1); \coordinate [above left=.5 of v0] (i0); \coordinate [above right=.5 of v0] (i1); \coordinate [right= of i1] (i11); \coordinate [below left=.5 of v1] (o0); \coordinate [below right=.5 of v1] (o1); \coordinate [right= of o1] (o11); \coordinate [above=.5 of v1] (vm); \draw (i0) -- (v0); \draw (o11) -- (v0); \draw (o0) -- (v1); \draw[color=white,line width=2pt] (i11) -- (v1); \draw (i11) -- (v1); \filldraw (v0) circle(1pt); \filldraw (v1) circle(1pt); \draw (v0) to[bend left=45] (v1); \draw (v0) to[bend right=45] (v1); \end{scope} \end{tikzpicture} } $ are isomorphic with respect to each other in the present scope. The symmetry factor of  ${ \def \scale {2ex} \begin{tikzpicture}[x=\scale,y=\scale,baseline={([yshift=-.5ex]current bounding box.center)}] \begin{scope}[node distance=1] \coordinate (v0); \coordinate [right=.5 of v0] (vm); \coordinate [right=.5 of vm] (v1); \coordinate [above left=.5 of v0] (i0); \coordinate [below left=.5 of v0] (i1); \coordinate [above right=.5 of v1] (o0); \coordinate [below right=.5 of v1] (o1); \draw (vm) circle(.5); \draw (i0) -- (v0); \draw (i1) -- (v0); \draw (o0) -- (v1); \draw (o1) -- (v1); \filldraw (v0) circle(1pt); \filldraw (v1) circle(1pt); \end{scope} \end{tikzpicture} }$ is $\frac{1}{|\Aut { \def \scale {2ex} \begin{tikzpicture}[x=\scale,y=\scale,baseline={([yshift=-.5ex]current bounding box.center)}] \begin{scope}[node distance=1] \coordinate (v0); \coordinate [right=.5 of v0] (vm); \coordinate [right=.5 of vm] (v1); \coordinate [above left=.5 of v0] (i0); \coordinate [below left=.5 of v0] (i1); \coordinate [above right=.5 of v1] (o0); \coordinate [below right=.5 of v1] (o1); \draw (vm) circle(.5); \draw (i0) -- (v0); \draw (i1) -- (v0); \draw (o0) -- (v1); \draw (o1) -- (v1); \filldraw (v0) circle(1pt); \filldraw (v1) circle(1pt); \end{scope} \end{tikzpicture} } |} = \frac{1}{16}$.

The labeling of the legs was omitted, because it makes the definition of the Hopf algebra of Feynman diagrams more involved without any gain in generality. 

Multiplying by the number of possible permutations of the external legs, which is $4!=24$ in the case of scalar four-leg diagrams, we reobtain the symmetry factors for the leg-fixed case: \begin{align} 4! \cdot \frac{1}{16} = \frac{3}{2} = \frac{1}{|\Aut_\text{lf} { \def \scale {2ex} \begin{tikzpicture}[x=\scale,y=\scale,baseline={([yshift=-.5ex]current bounding box.center)}] \begin{scope}[node distance=1] \coordinate (v0); \coordinate [right=.5 of v0] (vm); \coordinate [right=.5 of vm] (v1); \coordinate [above left=.5 of v0] (i0); \coordinate [below left=.5 of v0] (i1); \coordinate [above right=.5 of v1] (o0); \coordinate [below right=.5 of v1] (o1); \draw (vm) circle(.5); \draw (i0) -- (v0); \draw (i1) -- (v0); \draw (o0) -- (v1); \draw (o1) -- (v1); \filldraw (v0) circle(1pt); \filldraw (v1) circle(1pt); \end{scope} \end{tikzpicture} } |} + \frac{1}{|\Aut_\text{lf} { \def \scale {2ex} \begin{tikzpicture}[x=\scale,y=\scale,baseline={([yshift=-.5ex]current bounding box.center)}] \begin{scope}[node distance=1] \coordinate (v0); \coordinate [below=1 of v0] (v1); \coordinate [above left=.5 of v0] (i0); \coordinate [above right=.5 of v0] (i1); \coordinate [below left=.5 of v1] (o0); \coordinate [below right=.5 of v1] (o1); \coordinate [above=.5 of v1] (vm); \draw (vm) circle(.5); \draw (i0) -- (v0); \draw (i1) -- (v0); \draw (o0) -- (v1); \draw (o1) -- (v1); \filldraw (v0) circle(1pt); \filldraw (v1) circle(1pt); \end{scope} \end{tikzpicture} } |} + \frac{1}{|\Aut_\text{lf} { \def \scale {2ex} \begin{tikzpicture}[x=\scale,y=\scale,baseline={([yshift=-.5ex]current bounding box.center)}] \begin{scope}[node distance=1] \coordinate (v0); \coordinate [below=1 of v0] (v1); \coordinate [above left=.5 of v0] (i0); \coordinate [above right=.5 of v0] (i1); \coordinate [right= of i1] (i11); \coordinate [below left=.5 of v1] (o0); \coordinate [below right=.5 of v1] (o1); \coordinate [right= of o1] (o11); \coordinate [above=.5 of v1] (vm); \draw (i0) -- (v0); \draw (o11) -- (v0); \draw (o0) -- (v1); \draw[color=white,line width=2pt] (i11) -- (v1); \draw (i11) -- (v1); \filldraw (v0) circle(1pt); \filldraw (v1) circle(1pt); \draw (v0) to[bend left=45] (v1); \draw (v0) to[bend right=45] (v1); \end{scope} \end{tikzpicture} } |}, \end{align}     where $\Aut_\text{lf}$ are the automorphisms which leave the legs fixed. In the same way, we can always replace diagrams with non-fixed legs by a formal sum of diagrams with fixed legs and reobtain the usual notation.

\paragraph{Subdiagrams} A subdiagram of a Feynman diagram $\Gamma= (H,E,V,\eta)$ is a Feynman diagram $\gamma = (H',E',V',\eta')$ such that $E'\cup V' \subset E \cup V$ and $\eta' = \left. \eta \right|_{H'}$. The relation $\gamma$ is a subdiagram of $\Gamma$ is denoted as $\gamma \subset \Gamma$. 

The union and the intersection of two subdiagrams $\gamma_1=(H_1,E_1,V_1,\eta_1), \gamma_2=(H_2,E_2,V_2,\eta_2) \subset \Gamma=(H,E,V,\eta)$ is defined as \begin{align} \gamma_1 \cup \gamma_2 &= (H_1 \cup H_2, E_1 \cup E_2, V_1 \cup V_2, \left. \eta \right|_{H_1 \cup H_2}) \\ \gamma_1 \cap \gamma_2 &= (H_1 \cap H_2, E_1 \cap E_2, V_1 \cap V_2, \left. \eta \right|_{H_1 \cap H_2}). \end{align} $\gamma_1 \cup \gamma_2$ and $\gamma_1 \cap \gamma_2$ are also subdiagrams of $\Gamma$.

\paragraph{Connectedness} A diagram $\Gamma$ is \textit{connected} if it does not have two subdiagrams $\gamma_1$ and $\gamma_2$ such that $\gamma_1 \cup \gamma_2 = \Gamma$ and $\gamma_1 \cap \gamma_2 = \emptyset$, where $\emptyset$ is the diagram with no half-edges, $|H(\emptyset)| = 0$. The \textit{connected components} of a diagram are its connected subdiagrams which are mutually disjoint and whose union is the whole diagram.

The diagram $\Gamma$ is \textit{one-particle-irreducible} or \textit{1PI} if it is still connected after the removal of an arbitrary edge $e \in E(\Gamma)$.

\begin{align} \subdiags_{\text{1PI}}(\Gamma) := \left\{\gamma \subset \Gamma \text{ such that } \gamma \text{ is 1PI}\right\} \end{align} is the set of 1PI subdiagrams of $\Gamma$. If $\Gamma$ is 1PI, then $\Gamma \in \subdiags_{\text{1PI}}(\Gamma)$.

\begin{expl}[1PI subdiagrams of a diagram in $\phi^4$ theory] \label{expl:subdiagrams1PI} For the diagram  ${ \def \scale {2ex} 
 }$ in $\phi^4$ theory ( $\Phi = \left\{ \phi \right\}$, $\mathcal R_v = \left\{ \phi^4, \phi^2 \right\}$ and $\mathcal R_e =\left\{ \phi^2 \right\}$ ). \ifdefined\nodraft \def \thickness {2pt} \begin{gather*} \subdiags_{\text{1PI}} \left( { \def \scale {3ex} %
 } \right\}, \end{gather*} \else MISSING IN DRAFT MODE \fi where subdiagrams are drawn with thick lines. Note, that we do not need to fix external legs for the subdiagrams. Therefore, a subdiagram is determined entirely by its edge set. \end{expl} \paragraph{Singular homology and loop number} As a consequence of the singular homology of a Feynman diagram $\Gamma$, the first Betti number $h_1(\Gamma)$, which is often called \textit{loop number}, can be calculated. By the Euler characteristic, $h_1(\Gamma)- |E(\Gamma)| + |V(\Gamma)| - h_0(\Gamma) = 0$. Where $h_0(\Gamma)$, the zeroth Betti number, is the number of connected components of $\Gamma$. $h_1(\Gamma)$ is called loop number because it counts the number of independent circuits in the diagram. \paragraph{Contractions} A subdiagram $\gamma=(H',E',V',\eta')$ of a diagram $\Gamma=(H,E,V,\eta)$, $\gamma \subset \Gamma$ can be contracted. The result is denoted as $\Gamma / \gamma $. Explicitly: \begin{align} \label{eqn:contraction} \Gamma / \gamma &:= ( (H \setminus H') \cup H'_\text{ext}, E \setminus E', (V \setminus V') \cup \left\{ H'_\text{ext} \right\}, \eta'\big|_{(H \setminus H') \cup H'_\text{ext}} ). \end{align} The contraction of a subdiagram involves the removal of all the internal half-edges, all the edges and vertices of the subdiagram. Only the legs of the subdiagram are kept and `tied together' to form a new vertex of the resulting diagram. For a propagator-type subdiagram this means that an additional vertex with the same residue as the contracted graph is added. The new vertex in the produced diagram, $v^* \in \Gamma/\gamma$ has the residue as the contracted subdiagram, $\res(v^*) = \res(\gamma)$. For $\Gamma/\gamma$ to be a valid Feynman diagram with vertex-types in $\mathcal{R}_v$, the residues of contracted subdiagrams need to be restricted in a certain way. This will lead to the notion of superficial divergent subdiagrams. Moreover, this will explain the requirement $\mathcal{R}_e \subset \mathcal{R}_v$. \paragraph{Superficial degree of divergence} Using the map $\omega$, which is provided by the QFT, to assign a weight to every vertex and edge-type, an additional map $\omega_D$ can be defined, which assigns a weight to a Feynman diagram. This weight is called \textit{superficial degree of divergence} in the sense of \cite{Weinberg1960}: \begin{align} \label{eqn:omega_D} \omega_D\left(\Gamma\right) &:= \sum \limits_{e\in E} \omega(\text{\normalfont{res}}(e)) - \sum \limits_{v\in V} \omega(\text{\normalfont{res}}(v)) - D h_1 \left(\Gamma\right) \end{align} Neglecting possible infrared divergences, the value of $\omega_D$ coincides with the degree of divergence of the integral associated to the diagram in the perturbation expansion of the underlying QFT in $D$-dimensions. A 1PI diagram $\Gamma$ with $\omega_D(\Gamma) \leq 0$ is \textit{superficially divergent} (s.d.) in $D$ dimensions. For notational simplicity, the weight $0$ is assigned to the empty diagram, $\omega_D\left(\emptyset\right) = 0$, even though it is not divergent. \begin{defn}[Renormalizable Quantum Field Theory] \label{def:renormalizable} A QFT is \textit{renormalizable in $D$ dimensions} if $\omega_D(\Gamma)$ depends only on the external structure of $\Gamma$ and the superficial degree of divergence agrees with the weight assigned to the residue of the diagram: $\omega_D(\Gamma) = \omega(\res \Gamma)$. This can be expressed as the commutativity of the diagram: \begin{center} \begin{tikzpicture} \node (tl) {$\mathcal T$}; \node [right=of tl] (tr) {$\mathbb{Z}$}; \node [below =of tl] (bl) {$\mathbb{N}^\Phi$}; \draw[->] (tl) -- node[above] {$\omega_D$} (tr); \draw[->] (tl) to node[auto] {$\res$} (bl); \draw[->] (bl) -- node[right] {$\omega$} (tr); \end{tikzpicture} \end{center} \end{defn} where $\mathcal T$ is the set of all connected admissible Feynman diagrams of the renormalizable QFT. Specifically, $\omega_D(\Gamma)$ needs to be independent of $h_1(\Gamma)$. Working with a renormalizable QFT, we need to keep track of \textit{subdivergences} or \textit{superficially divergent subdiagrams} appearing in the integrals of the perturbation expansion. The tools needed are the set of 1PI subdiagrams and the superficial degree of divergence. The compatibility of the vertex and edge-weights and the superficial degree of divergence of the diagrams is exactly what is necessary to contract these subdivergences without leaving the space of allowed Feynman diagrams. \paragraph{Superficially divergent subdiagrams} The set of \textit{superficially divergent subdiagrams} or s.d. subdiagrams, \begin{align} \label{eqn:sdsubdiags} \sdsubdiags_D(\Gamma):= \left\{\gamma \subset \Gamma \text{ such that } \gamma = \prod \limits_i \gamma_i, \gamma_i \in \subdiags_{\text{1PI}}(\Gamma) \text{ and } \omega_D( \gamma_i ) \leq 0 \right\}, \end{align} of subdiagrams, whose connected components are s.d. 1PI diagrams, is the object of main interest for the combinatorics of renormalization. The renormalizability of the QFT guarantees that for every $\gamma \in \sdsubdiags_D(\Gamma)$ the diagram resulting from the contraction $\Gamma/\gamma$ is still a valid Feynman diagram of the underlying QFT. \begin{expl}[Superficially divergent subdiagrams of a diagram in $\phi^4$ theory] \label{expl:subdiagrams1PIsd} Consider the same diagram as in example \ref{expl:subdiagrams1PI} in $\phi^4$ theory with the weights $\omega(\phi^2) = \omega( \simpleprop ) = 2$ and $\omega(\phi^4) = \omega( \fourvtx ) = 0$. The superficially divergent subdiagrams for $D=4$ are \ifdefined\nodraft \def \thickness {2pt} \begin{gather*} \sdsubdiags_4 \left( { \def \scale {3ex} \begin{tikzpicture}[x=\scale,y=\scale,baseline={([yshift=-.5ex]current bounding box.center)}] \begin{scope}[node distance=1] \coordinate (v0); \coordinate[right=.5 of v0] (v4); \coordinate[above right= of v4] (v2); \coordinate[below right= of v4] (v3); \coordinate[below right= of v2] (v5); \coordinate[right=.5 of v5] (v1); \coordinate[above right= of v2] (o1); \coordinate[below right= of v2] (o2); \coordinate[below left=.5 of v0] (i1); \coordinate[above left=.5 of v0] (i2); \coordinate[below right=.5 of v1] (o1); \coordinate[above right=.5 of v1] (o2); \draw (v0) -- (i1); \draw (v0) -- (i2); \draw (v1) -- (o1); \draw (v1) -- (o2); \draw (v0) to[bend left=20] (v2); \draw (v0) to[bend right=20] (v3); \draw (v1) to[bend left=20] (v3); \draw (v1) to[bend right=20] (v2); \draw (v2) to[bend right=60] (v3); \draw (v2) to[bend left=60] (v3); \filldraw (v0) circle(1pt); \filldraw (v1) circle(1pt); \filldraw (v2) circle(1pt); \filldraw (v3) circle(1pt); \ifdefined\cvl \draw[line width=1.5pt] (v0) to[bend left=20] (v2); \draw[line width=1.5pt] (v0) to[bend right=20] (v3); \fi \ifdefined\cvr \draw[line width=1.5pt] (v1) to[bend left=20] (v3); \draw[line width=1.5pt] (v1) to[bend right=20] (v2); \fi \ifdefined\cvml \draw[line width=1.5pt] (v2) to[bend left=60] (v3); \fi \ifdefined\cvmr \draw[line width=1.5pt] (v2) to[bend right=60] (v3); \fi \end{scope} \end{tikzpicture} } \right) = \left\{ { \def \scale {3ex} \def \cvmr {} \def \cvml {} \begin{tikzpicture}[x=\scale,y=\scale,baseline={([yshift=-.5ex]current bounding box.center)}] \begin{scope}[node distance=1] \coordinate (v0); \coordinate[right=.5 of v0] (v4); \coordinate[above right= of v4] (v2); \coordinate[below right= of v4] (v3); \coordinate[below right= of v2] (v5); \coordinate[right=.5 of v5] (v1); \coordinate[above right= of v2] (o1); \coordinate[below right= of v2] (o2); \coordinate[below left=.5 of v0] (i1); \coordinate[above left=.5 of v0] (i2); \coordinate[below right=.5 of v1] (o1); \coordinate[above right=.5 of v1] (o2); \draw (v0) -- (i1); \draw (v0) -- (i2); \draw (v1) -- (o1); \draw (v1) -- (o2); \draw (v0) to[bend left=20] (v2); \draw (v0) to[bend right=20] (v3); \draw (v1) to[bend left=20] (v3); \draw (v1) to[bend right=20] (v2); \draw (v2) to[bend right=60] (v3); \draw (v2) to[bend left=60] (v3); \filldraw (v0) circle(1pt); \filldraw (v1) circle(1pt); \filldraw (v2) circle(1pt); \filldraw (v3) circle(1pt); \ifdefined\cvl \draw[line width=1.5pt] (v0) to[bend left=20] (v2); \draw[line width=1.5pt] (v0) to[bend right=20] (v3); \fi \ifdefined\cvr \draw[line width=1.5pt] (v1) to[bend left=20] (v3); \draw[line width=1.5pt] (v1) to[bend right=20] (v2); \fi \ifdefined\cvml \draw[line width=1.5pt] (v2) to[bend left=60] (v3); \fi \ifdefined\cvmr \draw[line width=1.5pt] (v2) to[bend right=60] (v3); \fi \end{scope} \end{tikzpicture} } , { \def \scale {3ex} \def \cvmr {} \def \cvml {} \def \cvl {} \begin{tikzpicture}[x=\scale,y=\scale,baseline={([yshift=-.5ex]current bounding box.center)}] \begin{scope}[node distance=1] \coordinate (v0); \coordinate[right=.5 of v0] (v4); \coordinate[above right= of v4] (v2); \coordinate[below right= of v4] (v3); \coordinate[below right= of v2] (v5); \coordinate[right=.5 of v5] (v1); \coordinate[above right= of v2] (o1); \coordinate[below right= of v2] (o2); \coordinate[below left=.5 of v0] (i1); \coordinate[above left=.5 of v0] (i2); \coordinate[below right=.5 of v1] (o1); \coordinate[above right=.5 of v1] (o2); \draw (v0) -- (i1); \draw (v0) -- (i2); \draw (v1) -- (o1); \draw (v1) -- (o2); \draw (v0) to[bend left=20] (v2); \draw (v0) to[bend right=20] (v3); \draw (v1) to[bend left=20] (v3); \draw (v1) to[bend right=20] (v2); \draw (v2) to[bend right=60] (v3); \draw (v2) to[bend left=60] (v3); \filldraw (v0) circle(1pt); \filldraw (v1) circle(1pt); \filldraw (v2) circle(1pt); \filldraw (v3) circle(1pt); \ifdefined\cvl \draw[line width=1.5pt] (v0) to[bend left=20] (v2); \draw[line width=1.5pt] (v0) to[bend right=20] (v3); \fi \ifdefined\cvr \draw[line width=1.5pt] (v1) to[bend left=20] (v3); \draw[line width=1.5pt] (v1) to[bend right=20] (v2); \fi \ifdefined\cvml \draw[line width=1.5pt] (v2) to[bend left=60] (v3); \fi \ifdefined\cvmr \draw[line width=1.5pt] (v2) to[bend right=60] (v3); \fi \end{scope} \end{tikzpicture} } , { \def \scale {3ex} \def \cvmr {} \def \cvml {} \def \cvr {} \begin{tikzpicture}[x=\scale,y=\scale,baseline={([yshift=-.5ex]current bounding box.center)}] \begin{scope}[node distance=1] \coordinate (v0); \coordinate[right=.5 of v0] (v4); \coordinate[above right= of v4] (v2); \coordinate[below right= of v4] (v3); \coordinate[below right= of v2] (v5); \coordinate[right=.5 of v5] (v1); \coordinate[above right= of v2] (o1); \coordinate[below right= of v2] (o2); \coordinate[below left=.5 of v0] (i1); \coordinate[above left=.5 of v0] (i2); \coordinate[below right=.5 of v1] (o1); \coordinate[above right=.5 of v1] (o2); \draw (v0) -- (i1); \draw (v0) -- (i2); \draw (v1) -- (o1); \draw (v1) -- (o2); \draw (v0) to[bend left=20] (v2); \draw (v0) to[bend right=20] (v3); \draw (v1) to[bend left=20] (v3); \draw (v1) to[bend right=20] (v2); \draw (v2) to[bend right=60] (v3); \draw (v2) to[bend left=60] (v3); \filldraw (v0) circle(1pt); \filldraw (v1) circle(1pt); \filldraw (v2) circle(1pt); \filldraw (v3) circle(1pt); \ifdefined\cvl \draw[line width=1.5pt] (v0) to[bend left=20] (v2); \draw[line width=1.5pt] (v0) to[bend right=20] (v3); \fi \ifdefined\cvr \draw[line width=1.5pt] (v1) to[bend left=20] (v3); \draw[line width=1.5pt] (v1) to[bend right=20] (v2); \fi \ifdefined\cvml \draw[line width=1.5pt] (v2) to[bend left=60] (v3); \fi \ifdefined\cvmr \draw[line width=1.5pt] (v2) to[bend right=60] (v3); \fi \end{scope} \end{tikzpicture} } , { \def \scale {3ex} \def \cvmr {} \def \cvml {} \def \cvl {} \def \cvr {} \begin{tikzpicture}[x=\scale,y=\scale,baseline={([yshift=-.5ex]current bounding box.center)}] \begin{scope}[node distance=1] \coordinate (v0); \coordinate[right=.5 of v0] (v4); \coordinate[above right= of v4] (v2); \coordinate[below right= of v4] (v3); \coordinate[below right= of v2] (v5); \coordinate[right=.5 of v5] (v1); \coordinate[above right= of v2] (o1); \coordinate[below right= of v2] (o2); \coordinate[below left=.5 of v0] (i1); \coordinate[above left=.5 of v0] (i2); \coordinate[below right=.5 of v1] (o1); \coordinate[above right=.5 of v1] (o2); \draw (v0) -- (i1); \draw (v0) -- (i2); \draw (v1) -- (o1); \draw (v1) -- (o2); \draw (v0) to[bend left=20] (v2); \draw (v0) to[bend right=20] (v3); \draw (v1) to[bend left=20] (v3); \draw (v1) to[bend right=20] (v2); \draw (v2) to[bend right=60] (v3); \draw (v2) to[bend left=60] (v3); \filldraw (v0) circle(1pt); \filldraw (v1) circle(1pt); \filldraw (v2) circle(1pt); \filldraw (v3) circle(1pt); \ifdefined\cvl \draw[line width=1.5pt] (v0) to[bend left=20] (v2); \draw[line width=1.5pt] (v0) to[bend right=20] (v3); \fi \ifdefined\cvr \draw[line width=1.5pt] (v1) to[bend left=20] (v3); \draw[line width=1.5pt] (v1) to[bend right=20] (v2); \fi \ifdefined\cvml \draw[line width=1.5pt] (v2) to[bend left=60] (v3); \fi \ifdefined\cvmr \draw[line width=1.5pt] (v2) to[bend right=60] (v3); \fi \end{scope} \end{tikzpicture} } , { \def \scale {3ex} \begin{tikzpicture}[x=\scale,y=\scale,baseline={([yshift=-.5ex]current bounding box.center)}] \begin{scope}[node distance=1] \coordinate (v0); \coordinate[right=.5 of v0] (v4); \coordinate[above right= of v4] (v2); \coordinate[below right= of v4] (v3); \coordinate[below right= of v2] (v5); \coordinate[right=.5 of v5] (v1); \coordinate[above right= of v2] (o1); \coordinate[below right= of v2] (o2); \coordinate[below left=.5 of v0] (i1); \coordinate[above left=.5 of v0] (i2); \coordinate[below right=.5 of v1] (o1); \coordinate[above right=.5 of v1] (o2); \draw (v0) -- (i1); \draw (v0) -- (i2); \draw (v1) -- (o1); \draw (v1) -- (o2); \draw (v0) to[bend left=20] (v2); \draw (v0) to[bend right=20] (v3); \draw (v1) to[bend left=20] (v3); \draw (v1) to[bend right=20] (v2); \draw (v2) to[bend right=60] (v3); \draw (v2) to[bend left=60] (v3); \filldraw (v0) circle(1pt); \filldraw (v1) circle(1pt); \filldraw (v2) circle(1pt); \filldraw (v3) circle(1pt); \ifdefined\cvl \draw[line width=1.5pt] (v0) to[bend left=20] (v2); \draw[line width=1.5pt] (v0) to[bend right=20] (v3); \fi \ifdefined\cvr \draw[line width=1.5pt] (v1) to[bend left=20] (v3); \draw[line width=1.5pt] (v1) to[bend right=20] (v2); \fi \ifdefined\cvml \draw[line width=1.5pt] (v2) to[bend left=60] (v3); \fi \ifdefined\cvmr \draw[line width=1.5pt] (v2) to[bend right=60] (v3); \fi \end{scope} \end{tikzpicture} } \right\}. \end{gather*} \else MISSING IN DRAFT MODE \fi \end{expl} \subsection{Hopf algebra structure of Feynman diagrams} The basis for the analysis of the lattice structure in QFTs is Kreimer's Hopf algebra of Feynman diagrams. It captures the BPHZ renormalization procedure which is necessary to obtain finite amplitudes from perturbative calculations in an algebraic framework \cite{ConnesKreimer2000}. In this section, the basic definitions of the Hopf algebra of Feynman diagrams will be repeated. For a more detailed exposition consult \cite{manchon2004} for mathematical details of Hopf algebras or \cite{borinsky2014feynman} for more computational aspects. \begin{defn} Let ${\hopffg_D}$ be the $\mathbb{Q}$-algebra generated by all mutually non-isomorphic Feynman diagrams, whose connected components are superficially divergent 1PI diagrams of a certain QFT in $D$-dimensions. The multiplication of generators is given by the disjoint union: $m: \hopffg_D \otimes \hopffg_D \rightarrow \hopffg_D, \gamma_1 \otimes \gamma_2 \mapsto \gamma_1 \cup \gamma_2$. It is extended linearly to all elements in the vector space $\hopffg_D$. ${\hopffg_D}$ has a unit $\unit:\mathbb{Q} \mapsto \mathbb{I} \mathbb{Q} \subset \hopffg_D $, where $\mathbb{I}$ is associated to the empty diagram and a counit $\counit: {\hopffg_D} \rightarrow \mathbb{Q}$, which vanishes on every generator of ${\hopffg_D}$ except $\mathbb{I}$: $\counit \circ \mathbb{I} := 1$. The coproduct on the generators is defined as follows: \begin{align} \label{eqn:def_cop} &\Delta_D \Gamma := \sum \limits_{ \gamma \in \sdsubdiags_D(\Gamma) } \gamma \otimes \Gamma/\gamma& &:& &{\hopffg_D} \rightarrow {\hopffg_D} \otimes {\hopffg_D}, \end{align} where the complete contraction $\Gamma/\Gamma$ is set to $\mathbb{I}$ and the vacuous contraction $\Gamma/\emptyset$ to $\Gamma$. The notion of superficial degree of divergence, $\omega_D$, hidden in $\sdsubdiags_D(\Gamma)$ (eq. \eqref{eqn:sdsubdiags}) is the only input to the Hopf algebra structure which depends on the dimension $D$ of spacetime. \end{defn} \begin{expl}[Coproduct of a diagram in $\phi^4$-theory] \label{expl:coproducthopffg} Take the same diagram of $\phi^4$-theory as in examples \ref{expl:subdiagrams1PI} and \ref{expl:subdiagrams1PIsd}. The coproduct is calculated using the set $\sdsubdiags_4(\Gamma)$ and the definition of the contraction in eq. \eqref{eqn:contraction}: \ifdefined\nodraft \begin{align*} \Delta_4 { \def \scale {2ex} 
 } \end{align*} The last equality holds because ${ \def \scale {1.5ex} %
 }$ are the same diagrams up to a permutation of the legs. \else MISSING IN DRAFT MODE \fi \end{expl} Extended linearly to all elements of $\hopffg_D$, the coproduct $\Delta_D: {\hopffg_D} \rightarrow {\hopffg_D} \otimes {\hopffg_D}$ is an algebra morphism. This implies $\Delta_D \mathbb{I} = \mathbb{I} \otimes \mathbb{I}$ and $\Delta_D \circ m = (m \otimes m) \circ \tau_{23} \circ ( \Delta_D \otimes \Delta_D )$, where $\tau_{23}$ switches the second and the third entry of the tensor product. The coproduct is coassociative: $(\Delta_D \otimes \id) \otimes \Delta_D = (\id \otimes \Delta_D) \otimes \Delta_D$. ${\hopffg_D}$ is graded by the loop number, $h_1(\Gamma)$, of the diagrams, \begin{align} {\hopffg_D} &= \bigoplus \limits_{L \geq 0} {\hopffg_D}^{(L)}\text{ and} \label{grading_loop_1} \\ m &: {\hopffg_D}^{(L_1)} \otimes {\hopffg_D}^{(L_2)} \rightarrow {\hopffg_D}^{(L_1 + L_2)} \label{grading_loop_2} \\ \Delta_D &: {\hopffg_D}^{(L)} \rightarrow \bigoplus \limits_{\substack{L_1, L_2 \ge 0 \\ L_1 + L_2 = L} } {\hopffg_D}^{(L_1)} \otimes {\hopffg_D}^{(L_2)}, \label{grading_loop_3} \end{align} where ${\hopffg_D}^{(L)} \subset {\hopffg_D}$ is the subspace of ${\hopffg_D}$ which is generated by diagrams $\Gamma$ with $h_1(\Gamma) = L$. Because ${\hopffg_D}$ is coassociative, it makes sense to define iterations of $\Delta_D$: \begin{align} \Delta_D^{0}&:= \counit, & \Delta_D^1 &:= \id ,& \Delta_D^{n} &:= ( \Delta_D \otimes \id^{\otimes n-2} ) \circ \Delta_D^{n-1} & \text{ for } n \geq 2, \end{align} where $\Delta_D^{k}: \hopffg_D \rightarrow {\hopffg_D}^{\otimes k}$. ${\hopffg_D}^{\otimes 0} \simeq \mathbb{Q}$, because the unit $\mathbb{I}$ is the only generator with $0$ loops. For this reason, $\hopffg_D$ is a connected Hopf algebra. The reduced coproduct is defined as $\widetilde{\Delta}_D:= P^{ \otimes 2} \circ \Delta_D $, where $P:= \id - \unit \circ \counit$ projects into the augmentation ideal, $P: {\hopffg_D} \rightarrow \bigoplus \limits_{L \geq 1} {\hopffg_D}^{(L)}$. \begin{expl}[Reduced coproduct of a non-primitive diagram in $\phi^4$-theory] \ifdefined\nodraft \begin{align} \widetilde \Delta_4 { \def \scale {2ex} \begin{tikzpicture}[x=\scale,y=\scale,baseline={([yshift=-.5ex]current bounding box.center)}] \begin{scope}[node distance=1] \coordinate (v0); \coordinate[right=.5 of v0] (v4); \coordinate[above right= of v4] (v2); \coordinate[below right= of v4] (v3); \coordinate[below right= of v2] (v5); \coordinate[right=.5 of v5] (v1); \coordinate[above right= of v2] (o1); \coordinate[below right= of v2] (o2); \coordinate[below left=.5 of v0] (i1); \coordinate[above left=.5 of v0] (i2); \coordinate[below right=.5 of v1] (o1); \coordinate[above right=.5 of v1] (o2); \draw (v0) -- (i1); \draw (v0) -- (i2); \draw (v1) -- (o1); \draw (v1) -- (o2); \draw (v0) to[bend left=20] (v2); \draw (v0) to[bend right=20] (v3); \draw (v1) to[bend left=20] (v3); \draw (v1) to[bend right=20] (v2); \draw (v2) to[bend right=60] (v3); \draw (v2) to[bend left=60] (v3); \filldraw (v0) circle(1pt); \filldraw (v1) circle(1pt); \filldraw (v2) circle(1pt); \filldraw (v3) circle(1pt); \ifdefined\cvl \draw[line width=1.5pt] (v0) to[bend left=20] (v2); \draw[line width=1.5pt] (v0) to[bend right=20] (v3); \fi \ifdefined\cvr \draw[line width=1.5pt] (v1) to[bend left=20] (v3); \draw[line width=1.5pt] (v1) to[bend right=20] (v2); \fi \ifdefined\cvml \draw[line width=1.5pt] (v2) to[bend left=60] (v3); \fi \ifdefined\cvmr \draw[line width=1.5pt] (v2) to[bend right=60] (v3); \fi \end{scope} \end{tikzpicture} } &= 2 { \def \scale {2ex} \begin{tikzpicture}[x=\scale,y=\scale,baseline={([yshift=-.5ex]current bounding box.center)}] \begin{scope}[node distance=1] \coordinate (v0); \coordinate[right=.5 of v0] (v4); \coordinate[above right= of v4] (v2); \coordinate[below right= of v4] (v3); \coordinate[above right=.5 of v2] (o1); \coordinate[below right=.5 of v3] (o2); \coordinate[below left=.5 of v0] (i1); \coordinate[above left=.5 of v0] (i2); \draw (v0) -- (i1); \draw (v0) -- (i2); \draw (v2) -- (o1); \draw (v3) -- (o2); \draw (v0) to[bend left=20] (v2); \draw (v0) to[bend right=20] (v3); \draw (v2) to[bend right=60] (v3); \draw (v2) to[bend left=60] (v3); \filldraw (v0) circle(1pt); \filldraw (v2) circle(1pt); \filldraw (v3) circle(1pt); \end{scope} \end{tikzpicture} } \otimes { \def \scale {2ex} \begin{tikzpicture}[x=\scale,y=\scale,baseline={([yshift=-.5ex]current bounding box.center)}] \begin{scope}[node distance=1] \coordinate (v0); \coordinate [right=.5 of v0] (vm); \coordinate [right=.5 of vm] (v1); \coordinate [above left=.5 of v0] (i0); \coordinate [below left=.5 of v0] (i1); \coordinate [above right=.5 of v1] (o0); \coordinate [below right=.5 of v1] (o1); \draw (vm) circle(.5); \draw (i0) -- (v0); \draw (i1) -- (v0); \draw (o0) -- (v1); \draw (o1) -- (v1); \filldraw (v0) circle(1pt); \filldraw (v1) circle(1pt); \end{scope} \end{tikzpicture} } + { \def \scale {2ex} \begin{tikzpicture}[x=\scale,y=\scale,baseline={([yshift=-.5ex]current bounding box.center)}] \begin{scope}[node distance=1] \coordinate (v0); \coordinate [below=1 of v0] (v1); \coordinate [above left=.5 of v0] (i0); \coordinate [above right=.5 of v0] (i1); \coordinate [below left=.5 of v1] (o0); \coordinate [below right=.5 of v1] (o1); \coordinate [above=.5 of v1] (vm); \draw (vm) circle(.5); \draw (i0) -- (v0); \draw (i1) -- (v0); \draw (o0) -- (v1); \draw (o1) -- (v1); \filldraw (v0) circle(1pt); \filldraw (v1) circle(1pt); \end{scope} \end{tikzpicture} } \otimes { \def \scale {2ex} \begin{tikzpicture}[x=\scale,y=\scale,baseline={([yshift=-.5ex]current bounding box.center)}] \begin{scope}[node distance=1] \coordinate (v0); \coordinate [right=.5 of v0] (vm1); \coordinate [right=.5 of vm1] (v1); \coordinate [right=.5 of v1] (vm2); \coordinate [right=.5 of vm2] (v2); \coordinate [above left=.5 of v0] (i0); \coordinate [below left=.5 of v0] (i1); \coordinate [above right=.5 of v2] (o0); \coordinate [below right=.5 of v2] (o1); \draw (vm1) circle(.5); \draw (vm2) circle(.5); \draw (i0) -- (v0); \draw (i1) -- (v0); \draw (o0) -- (v2); \draw (o1) -- (v2); \filldraw (v0) circle(1pt); \filldraw (v1) circle(1pt); \filldraw (v2) circle(1pt); \end{scope} \end{tikzpicture} } \end{align} \else MISSING IN DRAFT MODE \fi \end{expl} The kernel of the reduced coproduct, is the space of \textit{primitive} elements of the Hopf algebra, $\text{Prim} \hopffg_D := \ker\widetilde{\Delta}_D$. Primitive 1PI diagrams $\Gamma$ with $\Gamma \in \ker\widetilde{\Delta}_D$ are exactly those diagrams, which do not contain any subdivergences. They are also called \textit{skeleton diagrams}. More general, we can define $\widetilde{\Delta}_D^{n} = P^{\otimes n} \circ \Delta_D^{n}$. These morphisms give rise to an increasing filtration of $\hopffg_D$, the \textit{coradical filtration}: \begin{gather} \begin{aligned} \label{eqn:coradfilt} ~^{(n)}\hopffg_D &:= \ker \widetilde{\Delta}_D^{n+1} &\forall& n \geq 0 \end{aligned}\\ \mathbb{Q} \simeq ~^{(0)}\hopffg_D \subset ~^{(1)}\hopffg_D \subset \ldots \subset ~^{(n)}\hopffg_D \subset \ldots \subset \hopffg_D. \end{gather}

${\hopffg_D}$ is equipped with an \textit{antipode} $S_D: {\hopffg_D} \rightarrow {\hopffg_D}$, implicitly defined by the identity, $m \circ ( \text{id} \otimes S_D ) \circ \Delta_D = \unit \circ \counit$. Because $\hopffg_D$ is connected, there is always a unique solution for $S_D$, which can be calculated recursively.

\paragraph{Group of characters} The characters of $\hopffg_D$ are linear algebra morphisms from ${\hopffg_D}$ to some commutative, unital algebra $\mathcal{A}$. The set of all these morphisms forms a group and is denoted as $G^{{\hopffg_D}}_{\mathcal{A}}$. The product on this group, called the \textit{convolution product}, is given by $\phi * \psi = m_\mathcal{A} \circ ( \phi \otimes \psi ) \circ \Delta_D$ for $\phi, \psi \in G^{{\hopffg_D}}_{\mathcal{A}}$ and the unit of the group is the morphism $\unit_\mathcal{A} \circ \counit_{{\hopffg_D}}$.

\paragraph{Hopf ideals} A \textit{Hopf ideal} $I$ of a Hopf algebra $\mathcal{H}$ is an ideal of the algebra and a coideal of the coalgebra.  Explicitly, it is a subspace of a Hopf algebra, $I \subset \mathcal{H}$, such that  \begin{align} m ( I \otimes \mathcal{H} ) &\subset I \\ \Delta I &\subset I \otimes \mathcal{H} + \mathcal{H} \otimes I. \end{align} For every ideal $I$, the quotient Hopf algebra $\mathcal{H}/I$ can be defined, where two elements $x,y \in \mathcal{H}$ are equivalent if their difference is in the ideal $x-y \in I$. The quotient $\mathcal{H}/I$ is also a Hopf algebra.

\paragraph{Hopf algebra morphisms} Algebra morphisms which also respect the coalgebraic structure of the Hopf algebra are called \textit{Hopf algebra morphisms}. They are linear maps from a Hopf algebra $\mathcal{H}_1$ to another Hopf algebra $\mathcal{H}_2$. Explicitly, this means that $\psi : \mathcal{H}_1 \rightarrow \mathcal{H}_2$ is a Hopf algebra morphism iff $m_{\mathcal{H}_2} \circ ( \psi \otimes \psi ) = \psi \circ m_{\mathcal{H}_1}$ and $\Delta_{\mathcal{H}_2} \circ \psi = (\psi \otimes \psi) \circ \Delta_{\mathcal{H}_1}$. A Hopf algebra morphism as $\psi$ always gives rise to a Hopf ideal in $\mathcal{H}_1$. This ideal is the kernel of the morphism denoted as $\ker \psi$. The quotient Hopf algebra $\mathcal{H}_1/\ker\psi$ is isomorphic to $\mathcal{H}_2$.

\begin{figure} \ifdefined\nodraft \begin{center} {
\def\scale {4ex} \begin{tikzpicture}[x=\scale,y=\scale,baseline={([yshift=-.5ex]current bounding box.center)}] \begin{scope}[node distance=1] \coordinate (i); \coordinate[above right= of i] (v); \coordinate[above=.5 of v] (vm); \coordinate[above=.5 of vm] (vt); \coordinate[below right= of v] (o); \draw (vm) circle(.5); \draw (i) -- (v); \draw (o) -- (v); \filldraw (v) circle(1pt); \draw [fill=white] (vt) circle(1ex); \draw [fill=white,thick,pattern=north west lines] (vt) circle (1ex); \end{scope} \end{tikzpicture} }
\end{center} \else

MISSING IN DRAFT MODE

\fi \caption{Characteristic subdiagram of every tadpole diagram} \label{fig:snail} \end{figure}
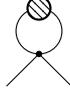

An useful example of a Hopf algebra morphism is the projection of $\hopffg_D$ to the Hopf algebra of Feynman diagrams without `tadpoles' (also snails or seagulls). Tadpoles are diagrams which can be split into two connected components by removing a single vertex such that one component does not contain any external leg. A tadpole diagram always has a subdiagram of a topology as depicted in fig. \ref{fig:snail}. The Hopf algebra of Feynman diagrams without tadpoles is denoted as $\hopffgs_D$. \begin{defn} \label{def:snailfreehopf} We define $\hopffgs_D$ as $\hopffg$ with the difference that no tadpole diagrams are allowed as generators and replace $\sdsubdiags_D(\Gamma)$ in the formula for the coproduct, eq. \eqref{eqn:def_cop}, with \begin{align} \sdsubdiagsn_D (\Gamma) := \left\{ \gamma \in \sdsubdiags_D(\Gamma) : \text{such that } \Gamma/\gamma \text{ is no tadpole diagram} \right\}. \end{align} \end{defn} Only the s.d.\ subdiagrams which do not result in a tadpole diagram upon contraction are elements of $\sdsubdiagsn_D (\Gamma)$.

A Hopf algebra morphism from $\hopffg_D$ to $\hopffgs_D$ is easy to set up: \begin{align} \label{eqn:morphismpsi} &\psi& &:& &\hopffg_D & &\rightarrow& &\hopffgs_D, \\ && && &\Gamma& &\mapsto& & \begin{cases} &0 \text{ if } \Gamma \text{ is a tadpole diagram.} \\ &\Gamma \text{ else.} \end{cases} \end{align} This map fulfills the requirements for a Hopf algebra morphism. The associated ideal $\ker\psi \subset \hopffg_D$ is the subspace of $\hopffg_D$ spanned by all tadpole diagrams. This ideal and the map $\psi$ are very useful, because the elements in $\ker \psi$ evaluate to zero after renormalization in kinematic subtraction schemes \cite{brown2013angles} and in minimal subtraction schemes for the massless case. \section{Algebraic lattice structure of subdivergences} \label{sec:posetsandlattices} To highlight the lattice structure of renormalizable QFTs, it is necessary to have more information on the structure in which subdivergences can appear. The classic notion of overlapping divergence, which addresses the case of two or more subdiagrams which are neither contained in each other nor disjoint, is central for this analysis.

\subsection{Overlapping divergences} \begin{defn} Two diagrams $\gamma_1, \gamma_2 \in \mathcal{P}_{\text{1PI}}(\Gamma)$ are called overlapping if $\gamma_1 \not \subset \gamma_2$, $\gamma_2 \not \subset \gamma_1$ and $\gamma_1 \cap \gamma_2 \neq \emptyset$.  \end{defn}

To prove the important properties of overlapping divergences exploited in the following sections, we need the following two lemmas: \begin{lmm} \label{lmm:sgidentities} For two arbitrary subdiagrams $\gamma_1, \gamma_2 \subset \Gamma$: \begin{align} \label{eqn:Hextid1} \Hext( \gamma_1 \cap \gamma_2 ) &= \left( \Hext( \gamma_1 ) \cup \Hext ( \gamma_2 ) \right) \cap H(\gamma_1) \cap H(\gamma_2) \\ \label{eqn:Hextid2} \Hext( \gamma_1 \cup \gamma_2 ) &= \left( \Hext( \gamma_1 ) \cup \Hext ( \gamma_2 ) \right) \setminus \Hint(\gamma_1) \setminus \Hint(\gamma_2) \\ | \Hext( \gamma_1 ) | + | \Hext( \gamma_2 ) | &= | \Hext( \gamma_1 \cup \gamma_2 ) | + | \Hext( \gamma_1 \cap \gamma_2 ) |. \end{align} \end{lmm} \begin{proof} We abbreviate $H(\gamma_1) =: H^1$, $\Hint(\gamma_2)=:H_i^2$, $\Hext( \gamma_1 \cup \gamma_2 ) =: H_e^{1 \cup 2}$, $\Hext( \gamma_1 \cap \gamma_2 )=: H_e^{1 \cap 2}$ and so on. The first and the second identity follow from the definition of $\Hext(\gamma)$ and the condition that edges of a diagram $\Gamma$ are either disjoint or equal: \begin{align} H_e^{1 \cap 2} & = \left( H^1 \cap H^2 \right) \setminus \left( \bigcup \limits_{e\in E(\gamma_1) \cap E(\gamma_2)} e \right) \end{align} using condition \ref{cond:disjedges} of definition \ref{def:feynmandiagram}, \begin{align} = \left( H^1 \cap H^2 \right) \setminus \left( H_i^1 \cap H_i^2 \right) = \left( H_e^1 \cap H^2 \right) \cup \left( H_e^2 \cap H^1 \right) = \left( H_e^1 \cup H_e^2 \right) \cap H^1 \cap H^2. \end{align} Analogously, the second identity follows. The third identity follows from the other two by applying inclusion-exclusion on the right-hand side of eq. \eqref{eqn:Hextid1} and eq. \eqref{eqn:Hextid2},  $| H_e^{1 \cap 2} | = | H_e^1 \cap H^2 | + | H_e^2 \cap H^1 | - | H_e^1 \cap H_e^2 |$,  $| H_e^{1 \cup 2} | = | H_e^1 \setminus H_i^2 | + | H_e^2 \setminus H_i^1 | - | H_e^1 \cap H_e^2 |$,  adding both and using inclusion-exclusion again: \begin{align} \begin{split} | H_e^{1 \cap 2} | + | H_e^{1 \cup 2} | &= | (H_e^1 \cap H^2 ) \cup ( H_e^1 \setminus H_i^2 ) | + | (H_e^1 \cap H^2 ) \cap ( H_e^1 \setminus H_i^2 ) | \\ &+ | (H_e^2 \cap H^1 ) \cup ( H_e^2 \setminus H_i^1 ) | + | (H_e^2 \cap H^1 ) \cap ( H_e^2 \setminus H_i^1 ) | \\ &- 2 | H_e^1 \cap H_e^2 | = | H_e^1 | + | H_e^2 | \end{split} \end{align} \end{proof}

\begin{lmm} \label{lmm:sdHext} For every proper subdiagram $\gamma$ of a connected diagram $\Gamma$: \begin{align} |\Hext(\gamma) \cap \Hint(\Gamma) | \geq 1 \end{align} \end{lmm} \begin{proof} Suppose $|\Hext(\gamma) \cap \Hint(\Gamma) | = 0$, that means there are no external legs of $\gamma$, which are internal in $\Gamma$. As a consequence, all half-edges of $\gamma$ can be removed from $\Gamma$ without breaking an edge. That means, $\Gamma \setminus \gamma=(H(\Gamma) \setminus H(\gamma), E(\Gamma) \setminus E(\gamma), V(\Gamma) \setminus V(\gamma), \eta \big|_{H(\Gamma) \setminus H(\gamma)} )$ is also a Feynman diagram. But $(\Gamma \setminus \gamma) \cup \gamma = \Gamma$ and $(\Gamma \setminus \gamma) \cap \gamma = \emptyset$, thus $\Gamma$ must be disconnected. \end{proof} \begin{crll} \label{crll:1PIsdHext} For every proper subdiagram $\gamma$ of a 1PI diagram $\Gamma$: \begin{align} |\Hext(\gamma) \cap \Hint(\Gamma) | \geq 2 \end{align} \end{crll} \begin{proof} Suppose $|\Hext(\gamma) \cap \Hint(\Gamma) | = 1$ and let $h \in \Hext(\gamma) \cap \Hint(\Gamma)$. As $h$ belongs to some edge $e \in E$ of $\Gamma=(H, E, V, \eta)$ such that $h \in e$ , this edge can be removed from $\Gamma$. Let $\Gamma' = (H, E \setminus e, V, \eta)$ which is still a connected diagram, because $\Gamma$ is 1PI. Accordingly, $|\Hext(\gamma) \cap \Hint(\Gamma') | = 0$, which contradicts the result of lemma \ref{lmm:sdHext}. \end{proof}

\begin{prop} \label{prop:union_4ext} If $\gamma_1$ and $\gamma_2$ are overlapping subdiagrams in $\mathcal{P}_{\text{1PI}}(\Gamma)$, and $\mu \subset \gamma_1 \cap \gamma_2$ some connected component of $\gamma_1 \cap \gamma_2$, then $|H_\text{ext}(\mu)| \geq 4$. \end{prop} \begin{proof} $\gamma_1$ and $\gamma_2$ are overlapping and $\mu$ is a proper subdiagram of both. In this case corollary \ref{crll:1PIsdHext} applies: $|\Hext(\mu) \cap \Hint(\gamma_1)| \geq 2$ and $|\Hext(\mu) \cap \Hint(\gamma_2)| \geq 2$. It remains to be proven that the sets $\Hext(\mu) \cap \Hint(\gamma_1)$ and $\Hext(\mu) \cap \Hint(\gamma_2)$ are mutually disjoint. $\mu$ is a connected component of $\gamma_1 \cap \gamma_2$, consequently $\Hext(\mu) \subset \Hext(\gamma_1 \cap \gamma_2)$. Lemma \ref{lmm:sgidentities} implies $\Hext(\gamma_1 \cap \gamma_2) \subset \Hext(\gamma_1) \cup \Hext(\gamma_2) \subset \overline{ \Hint(\gamma_1) \cap \Hint(\gamma_2) }$. Hence, $\Hext(\mu) \cap \Hint(\gamma_1) \cap \Hint(\gamma_2) = \emptyset$.  \end{proof} \subsection{Posets and algebraic lattices} The set of subdivergences of a Feynman diagram is obviously partially ordered by inclusion. These posets are quite constrained for some renormalizable QFTs: They additionally carry a lattice structure. In \cite[Part III]{figueroa2005combinatorial} this was studied in the light of distributive lattices. 

In this section, we will elaborate on the conditions a QFT must fulfill for these posets to be lattices. The term \textit{join-meet-renormalizability} will be defined which characterizes QFTs in which all Feynman diagrams whose set of subdivergencies form lattices. It will be shown that this is a special property of QFTs with only four-or-less-valent vertices. A counter example of a Feynman diagram with six-valent vertices which does not carry a lattice structure is given.

The definitions will be illustrated by the application to the set of subdivergences of a Feynman diagram. Additionally, we will introduce the corresponding Hopf algebra for these lattices based on an incidence Hopf algebra \cite{Schmitt1994}.

First, the necessary definitions of poset and lattice theory will be repeated:

\begin{defn}[Poset] A \textit{partially ordered set} or \textit{poset} is a finite set $P$ endowed with a partial order $\leq$.  An interval $[x,y]$ is a subset $\left\{ z \in P : x \leq z \leq y \right\} \subset P$. If $[x,y] = \left\{x,y\right\}$, x covers y and y is covered by x. \end{defn} For a more detailed exposition of poset and lattice theory we refer to \cite{stanley1997}.

\paragraph{Hasse diagram} A Hasse diagram of a poset $P$ is the graph with the elements of $P$ as vertices and the cover relations as edges. Larger elements are always drawn above smaller elements. 

\begin{expl} The set of superficially divergent subdiagrams $\sdsubdiags_D(\Gamma)$ of a Feynman diagram $\Gamma$ is a poset ordered by inclusion:  $\gamma_1 \leq \gamma_2 \Leftrightarrow \gamma_1 \subset \gamma_2$ for all $\gamma_1, \gamma_2 \in \sdsubdiags_D(\Gamma)$.

The statement that a subdiagram $\gamma_1$ covers $\gamma_2$ in $\sdsubdiags_D(\Gamma)$ is equivalent to the statement that $\gamma_1 / \gamma_2$ is primitive. The elements that are covered by the full diagram $\Gamma \in \sdsubdiags_D(\Gamma)$ are called \textit{maximal forests}; whereas, a maximal chain $\emptyset \subset \gamma_1 \subset \ldots \subset \gamma_n \subset \Gamma$, where each element is covered by the next, is a \textit{complete forest} of $\Gamma$. \end{expl}

The Hasse diagram of a s.d.\ diagram $\Gamma$ can be constructed by the following procedure: Draw the diagram and find all the maximal forests $\gamma_i \in \sdsubdiags_D(\Gamma)$ such that $\Gamma/\gamma_i$ is primitive. Draw the diagrams $\gamma_i$ under $\Gamma$ and draw lines from $\Gamma$ to the $\gamma_i$. Subsequently, determine all the maximal forests $\mu_i$ of the $\gamma_i$ and draw them under the $\gamma_i$. Draw a line from $\gamma_i$ to $\mu_i$ if $\mu_i \subset \gamma_i$. Repeat this until only primitive diagrams are left. Then draw lines from the primitive subdiagrams to an additional $\emptyset$-diagram underneath them. In the end, replace diagrams by vertices.

\begin{expl} For instance, the set of superficially divergent subdiagrams for $D=4$ of the diagram, \ifdefined\nodraft $ { \def\scale{2ex} \begin{tikzpicture}[x=\scale,y=\scale,baseline={([yshift=-.5ex]current bounding box.center)}] \begin{scope}[node distance=1] \coordinate (v0); \coordinate[right=.5 of v0] (v4); \coordinate[above right= of v4] (v2); \coordinate[below right= of v4] (v3); \coordinate[below right= of v2] (v5); \coordinate[right=.5 of v5] (v1); \coordinate[above right= of v2] (o1); \coordinate[below right= of v2] (o2); \coordinate[below left=.5 of v0] (i1); \coordinate[above left=.5 of v0] (i2); \coordinate[below right=.5 of v1] (o1); \coordinate[above right=.5 of v1] (o2); \draw (v0) -- (i1); \draw (v0) -- (i2); \draw (v1) -- (o1); \draw (v1) -- (o2); \draw (v0) to[bend left=20] (v2); \draw (v0) to[bend right=20] (v3); \draw (v1) to[bend left=20] (v3); \draw (v1) to[bend right=20] (v2); \draw (v2) to[bend right=60] (v3); \draw (v2) to[bend left=60] (v3); \filldraw (v0) circle(1pt); \filldraw (v1) circle(1pt); \filldraw (v2) circle(1pt); \filldraw (v3) circle(1pt); \ifdefined\cvl \draw[line width=1.5pt] (v0) to[bend left=20] (v2); \draw[line width=1.5pt] (v0) to[bend right=20] (v3); \fi \ifdefined\cvr \draw[line width=1.5pt] (v1) to[bend left=20] (v3); \draw[line width=1.5pt] (v1) to[bend right=20] (v2); \fi \ifdefined\cvml \draw[line width=1.5pt] (v2) to[bend left=60] (v3); \fi \ifdefined\cvmr \draw[line width=1.5pt] (v2) to[bend right=60] (v3); \fi \end{scope} \end{tikzpicture} }$ can be represented as the Hasse diagram $ { \def\scale{2ex} \begin{tikzpicture}[x=\scale,y=\scale,baseline={([yshift=-.5ex]current bounding box.center)}] \begin{scope}[node distance=1] \coordinate (top) ; \coordinate [below left= of top] (v1); \coordinate [below right= of top] (v2); \coordinate [below left= of v2] (v3); \coordinate [below= of v3] (bottom); \draw (top) -- (v1); \draw (top) -- (v2); \draw (v1) -- (v3); \draw (v2) -- (v3); \draw (v3) -- (bottom); \filldraw[fill=white, draw=black] (top) circle(2pt); \filldraw[fill=white, draw=black] (v1) circle(2pt); \filldraw[fill=white, draw=black] (v2) circle(2pt); \filldraw[fill=white, draw=black] (v3) circle(2pt); \filldraw[fill=white, draw=black] (bottom) circle(2pt); \end{scope} \end{tikzpicture} } $ \else

MISSING IN DRAFT MODE

\fi , where the vertices represent the subdiagrams in the set given in example \ref{expl:subdiagrams1PIsd}. \end{expl}

\begin{defn}[Lattice] A lattice is a poset $L$ for which an unique least upper bound (\textit{join}) and an unique greatest lower bound (\textit{meet}) exists for any combination of two elements in $L$. The join of two elements $x,y \in L$ is denoted as $x \vee y$ and the meet as $x \wedge y$. Every lattice has a unique greatest element denoted as $\hat{1}$ and a unique smallest element $\hat{0}$. Every interval of a lattice is also a lattice. \end{defn} In many QFTs $\sdsubdiags_D(\Gamma)$ is a lattice for every s.d.\ diagram $\Gamma$: \begin{defn}[Join-meet-renormalizable quantum field theory] A renormalizable QFT is called join-meet-renormalizable if $\sdsubdiags_D(\Gamma)$, ordered by inclusion, is a lattice for every s.d.\ Feynman diagram $\Gamma$. \end{defn} \begin{thm} A renormalizable QFT is join-meet-renormalizable if $\sdsubdiags_D(\Gamma)$ is closed under taking unions: $\gamma_1, \gamma_2 \in \sdsubdiags_D(\Gamma) \Rightarrow \gamma_1 \cup \gamma_2 \in \sdsubdiags_D(\Gamma)$ for all s.d.\ diagrams $\Gamma$. \end{thm} \begin{proof} $\sdsubdiags_D(\Gamma)$ is ordered by inclusion $\gamma_1 \leq \gamma_2 \Leftrightarrow \gamma_1 \subset \gamma_2$. The join is given by taking the union of diagrams: $\gamma_1 \join \gamma_2 := \gamma_1 \cup \gamma_2$. $\sdsubdiags_D(\Gamma)$ has a unique greatest element $\hat{1} := \Gamma$ and a unique smallest element $\hat{0} := \emptyset$. Therefore $\sdsubdiags_D(\Gamma)$ is a lattice \cite[Prop. 3.3.1]{stanley1997}. The unique meet is given by the formula, $\gamma_1 \meet \gamma_2 := \bigcup \limits_{\mu \leq \gamma_1 \text{ and } \mu \leq \gamma_2} \mu$.  \end{proof}

A broad class of renormalizable QFTs is join-meet-renormalizable including the standard model: \begin{thm} If all diagrams with four or more legs in a renormalizable QFT are superficially logarithmic divergent or superficially convergent, then the underlying QFT is join-meet-renormalizable. \end{thm} \begin{proof} From $\gamma_1, \gamma_2 \in \sdsubdiags_D(\Gamma)$ immediately follows $\gamma_1 \cup \gamma_2 \in \sdsubdiags_D(\Gamma)$ if $\gamma_1$ and $\gamma_2$ are disjoint or contained in each other.  The statement only needs to be validated if $\gamma_1$ and $\gamma_2$ are overlapping. In this case, we have $\gamma_1 \cup \gamma_2 \in \subdiags_{\text{1PI}}(\Gamma)$.

From eq. \eqref{eqn:omega_D}, the definition of $\omega_D$, and inclusion-exclusion: \begin{align} \omega_D(\gamma_1 \cup \gamma_2) &= \omega_D(\gamma_1) + \omega_D(\gamma_2) - \omega_D(\gamma_1 \cap \gamma_2) \end{align} $\gamma_1 \cap \gamma_2$ has four or more external legs as was proven in proposition \ref{prop:union_4ext}. Consequently, $\omega_D(\gamma_1 \cap \gamma_2) \geq 0$ and if $\omega_D(\gamma_1)\leq0$ and $\omega_D(\gamma_2)\leq 0$, then $\omega_D(\gamma_1 \cup \gamma_2) \leq 0$. For this reason $\gamma_1 \cup \gamma_2$ is superficially divergent, $\gamma_1 \cup \gamma_2 \in \sdsubdiags_D(\Gamma)$ and $\sdsubdiags_D(\Gamma)$ is closed under taking unions. \end{proof} \begin{crll} All renormalizable QFTs with only four-or-less-valent vertices are join-meet-renormalizable. \end{crll} \begin{proof} If there are only four-or-less-valent vertices, the diagrams with $|\Hext(\Gamma)| > 4$ must be superficially convergent. This also implies that $\omega_D(\Gamma) \geq 0$ for $|\Hext(\Gamma)|\geq 4$. \end{proof}

In general, renormalizable QFTs are not join-meet-renormalizable. Fig. \ref{fig:nolattice} shows an example of a s.d.\ diagram $\Gamma$, where $\sdsubdiags_D(\Gamma)$ is not a lattice . The diagram is depicted in fig. \ref{fig:phi6nolattice} and the corresponding poset in fig. \ref{fig:phi6nolatticeposet}. The diagram appears in $\phi^6$-theory, which is therefore renormalizable, but not join-meet-renormalizable, in $3$-dimensions.  {
\ifdefined\nodraft \begin{figure}   \subcaptionbox{Example of a diagram where $\sdsubdiags_3(\Gamma)$ is not a lattice.\label{fig:phi6nolattice}}   [.45\linewidth]{ \def \scale{2em} \begin{tikzpicture}[x=\scale,y=\scale,baseline={([yshift=-.5ex]current bounding box.center)}] \begin{scope}[node distance=1] \coordinate (v0); \coordinate[right=of v0] (v4); \coordinate[above right= of v4] (v2); \coordinate[below right= of v4] (v3); \coordinate[below right= of v2] (v5); \coordinate[right=of v5] (v1); \coordinate[above right= of v2] (o1); \coordinate[below right= of v2] (o2); \coordinate[below left=.5 of v0] (i1); \coordinate[above left=.5 of v0] (i2); \coordinate[below right=.5 of v1] (o1); \coordinate[above right=.5 of v1] (o2); \draw (v0) -- (i1); \draw (v0) -- (i2); \draw (v1) -- (o1); \draw (v1) -- (o2); \draw (v0) to[bend left=20] (v2); \draw (v0) to[bend left=45] (v2); \draw (v0) to[bend right=45] (v3); \draw (v0) to[bend right=20] (v3); \draw (v1) to[bend left=20] (v3); \draw (v1) to[bend left=45] (v3); \draw (v1) to[bend right=45] (v2); \draw (v1) to[bend right=20] (v2); \draw (v2) to[bend right=20] (v4); \draw (v2) to[bend left=20] (v5); \draw (v3) to[bend right=20] (v5); \draw (v3) to[bend left=20] (v4); \draw (v4) to[bend left=45] (v5); \draw (v4) to[bend left=15] (v5); \draw (v4) to[bend right=45] (v5); \draw (v4) to[bend right=15] (v5); \filldraw (v0) circle(1pt); \filldraw (v1) circle(1pt); \filldraw (v2) circle(1pt); \filldraw (v3) circle(1pt); \filldraw (v4) circle(1pt); \filldraw (v5) circle(1pt); \end{scope} \end{tikzpicture}   }   \subcaptionbox{The corresponding non-lattice poset. Trivial vertex multiplicities were omitted. \label{fig:phi6nolatticeposet}}   [.45\linewidth]{ \def \scale{1em} \begin{tikzpicture}[x=\scale,y=\scale,baseline={([yshift=-.5ex]current bounding box.center)}] \begin{scope}[node distance=1] \coordinate (top) ; \coordinate [below left=2 and 2 of top] (v1) ; \coordinate [below left=2 and 1 of top] (a2) ; \coordinate [below =2 of top] (a4) ; \coordinate [below =1 of a2] (v2) ; \coordinate [below =1 of a4] (v4) ; \coordinate [below right= 2 and 1 of top] (v5) ; \coordinate [below left = 4 and 1 of top] (w1) ; \coordinate [below= 4 of top] (w2) ; \coordinate [below left = 5 and .5 of top] (u1) ; \coordinate [below right =3 and 3 of top] (u2) ; \coordinate [below = 6 of top] (s) ; \coordinate [below = 7 of top] (t) ; \draw (top) -- (v1); \draw (top) -- (a2); \draw (top) -- (u2); \draw (top) -- (a4); \draw (top) -- (v5); \draw (a2) -- (v2); \draw (a4) -- (v4); \draw[color=white,line width=4pt] (v4) -- (w1); \draw (v4) -- (w1); \draw[color=white,line width=4pt] (v2) -- (w2); \draw (v2) -- (w1); \draw (v2) -- (w2); \draw (v4) -- (w2); \draw (v1) -- (w1); \draw (v5) -- (w2); \draw[color=white,line width=4pt] (w1) -- (u1); \draw[color=white,line width=4pt] (w2) -- (u1); \draw (w1) -- (u1); \draw (w2) -- (u1); \draw (u1) -- (s); \draw (u2) -- (s); \draw (s) -- (t); \filldraw[fill=white, draw=black] (top) circle(2pt); \filldraw[fill=white, draw=black] (v1) circle(2pt); \filldraw[fill=white, draw=black] (v2) circle(2pt); \filldraw[fill=white, draw=black] (v4) circle(2pt); \filldraw[fill=white, draw=black] (v5) circle(2pt); \filldraw[fill=white, draw=black] (a2) circle(2pt); \filldraw[fill=white, draw=black] (a4) circle(2pt); \filldraw[fill=white, draw=black] (w1) circle(2pt); \filldraw[fill=white, draw=black] (w2) circle(2pt); \filldraw[fill=white, draw=black] (u1) circle(2pt); \filldraw[fill=white, draw=black] (u2) circle(2pt); \filldraw[fill=white, draw=black] (s) circle(2pt); \filldraw[fill=white, draw=black] (t) circle(2pt); \end{scope} \end{tikzpicture}   }   \caption{Counter-example for a renormalizable but not join-meet-renormalizable QFT: $\phi^6$-theory in $3$ dimensions.}   \label{fig:nolattice} \end{figure}
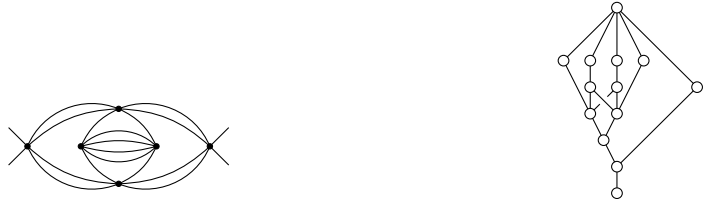 \else

MISSING IN DRAFT MODE

\fi }

To proceed to the Hopf algebra of decorated posets some additional notation of poset and lattice theory must be introduced:

\paragraph{Order preserving maps} A map $\sigma: P \rightarrow \mathbb{N}_0$ on a poset to the non-negative numbers is called strictly order preserving  if $x < y$ implies $\sigma(x) < \sigma(y)$ for all $x,y \in P$.  \paragraph{Cartesian product of posets} From two posets $P_1$ and $P_2$ a new poset $P_1 \times P_2 = \left\{ (s,t) : s \in P_1 \text{ and } t \in P_2 \right\}$, the Cartesian product, with the order relation, $(s,t) \leq (s',t')$ iff $s \leq s'$ and $t \leq t'$, is obtained.

The Cartesian product is commutative and if $P_1$ and $P_2$ are lattices $P_1 \times P_2$ is also a lattice \cite{stanley1997}. It is compatible with the notion of intervals: \begin{align} P_1 \times P_2 \supset [(s,t), (s',t')] =\left\{ (x,y) \in P_1 \times P_2 : s \leq x\leq s' \land t \leq y \leq t' \right\} = [s,s'] \times [t,t']. \end{align}

\paragraph{Isomorphisms of posets} An isomorphism between two posets $P_1$ and $P_2$ is a bijection $j: P_1 \rightarrow P_2$, which preserves the order relation: $j(x) \leq j(y) \Leftrightarrow x \leq y$. 

\subsection{The Hopf algebra of decorated posets} Using the preceding notions a new Hopf algebra structure on posets, suitable for the description of the subdivergences, can be defined. This structure is essentially the one of an incidence Hopf algebra \cite{Schmitt1994} augmented by a strictly order preserving map as a decoration. This is a standard procedure as most applications of posets and lattices require an combinatorial interpretation of the elements of the posets \cite{stanley1997} - analogous to the applications of the Hopf algebras \cite{joni1979coalgebras}. \begin{defn}[Hopf algebra of decorated posets] Let $\mathcal{D}$ be the set of tuples $(P, \nu)$, where $P$ is a finite poset with a unique lower bound $\hat{0}$ and a unique upper bound $\hat{1}$ and a strictly order preserving map $\nu: P \rightarrow \mathbb{N}_0$ with $\nu(\hat{0}) = 0$.  One can think of $\mathcal{D}$ as the set of bounded posets augmented by a strictly order preserving decoration. An equivalence relation is set up on $\mathcal{D}$ by relating $(P_1, \nu_1) \sim (P_2, \nu_2)$ if there is an isomorphism $j: P_1 \rightarrow P_2$, which respects the decoration $\nu$: $\nu_1 = \nu_2 \circ j$. 

Let $\hopfpos$ be the $\mathbb{Q}$-algebra generated by all the elements in the quotient $\mathcal{P}/\sim$ with the commutative multiplication:  \begin{align} &m_{\hopfpos}: & \hopfpos &\otimes \hopfpos & &\rightarrow& &\hopfpos, \\ && (P_1, \nu_1) &\otimes (P_2, \nu_2) & &\mapsto& &\left( P_1 \times P_2, \nu_1 + \nu_2 \right), \end{align} which takes the Cartesian product of the two posets and adds the decorations $\nu$. The sum of the two functions $\nu_1$ and $\nu_2$ is to be interpreted in the sense: $(\nu_1 + \nu_2)(x,y) = \nu_1(x) + \nu_2(y)$. The singleton poset $P=\left\{\hat{0}\right\}$ with $\hat{0}=\hat{1}$ and the trivial decoration $\nu(\hat{0}) = 0$ serves as a multiplicative unit: $\unit(1) = \mathbb{I}_{\hopfpos} := (\left\{\hat{0}\right\}, \hat{0} \mapsto 0)$.

Equipped with the coproduct, \begin{align} \label{eqn:poset_cop} &\Delta_{\hopfpos}: & &\hopfpos& &\rightarrow& \hopfpos &\otimes \hopfpos, \\ && &(P, \nu)& &\mapsto& \sum \limits_{ x \in P } ( [ \hat{0}, x ], \nu ) &\otimes \left( [x, \hat{1}], \nu - \nu(x) \right), \end{align} where $(\nu - \nu(x))(y) = \nu(y) - \nu(x)$  and the counit $\counit$ which vanishes on every generator except $\mathbb{I}_{\hopfpos}$, the algebra $\hopfpos$ becomes a counital coalgebra.  \end{defn} \begin{prop} $\hopfpos$ is a bialgebra. \end{prop} \begin{proof} The compatibility of the multiplication with the coproduct needs to be proven. Let $(P_1, \nu_1), (P_2, \nu_2) \in \mathcal{P}$. \begin{align} \begin{split} &\Delta_{\hopfpos} \circ m_{\hopfpos} ( (P_1, \nu_1) \otimes (P_2, \nu_2) ) = \Delta_{\hopfpos} \left( P_1 \times P_2, \nu_1 + \nu_2 \right) = \\ &\sum \limits_{ x \in P_1 \times P_2 } ( [ \hat{0}, x ], \nu_1 + \nu_2 ) \otimes \left( [x, \hat{1}], \nu_1 + \nu_2 - \nu_1(x) - \nu_2(x) \right) = \\ &\sum \limits_{ y \in P_1 } \sum \limits_{ z \in P_2 }( [ \hat{0}_{P_1}, y ] \times [ \hat{0}_{P_2}, z ], \nu_1 + \nu_2 ) \otimes \left( [y, \hat{1}_{P_1}] \times [z, \hat{1}_{P_2}], \nu_1 + \nu_2 - \nu_1(x) - \nu_2(x) \right) = \\ &(m_{\hopfpos} \otimes m_{\hopfpos}) \circ \sum \limits_{ y \in P_1 } \sum \limits_{ z \in P_2 } \big[ ( [ \hat{0}_{P_1}, y ], \nu_1) \otimes ( [ \hat{0}_{P_2}, z ], \nu_2 ) \\ &\otimes ( [y, \hat{1}_{P_1}], \nu_1 - \nu_1(x) ) \otimes ( [z, \hat{1}_{P_2}], \nu_2 - \nu_2(x) ) \big] = \\ &(m_{\hopfpos} \otimes m_{\hopfpos}) \circ \tau_{2,3} \circ (\Delta_{\hopfpos} \otimes \Delta_{\hopfpos}) ( (P_1, \nu_1) \otimes (P_2, \nu_2) ), \end{split} \end{align} where $\tau_{2,3}$ switches the second and the third factor of the tensor product. \end{proof} Note, that we also could have decorated the \textit{covers} of the lattices instead of the elements. We would have obtained a construction as in \cite{bergeron1999hopf} with certain restrictions on the edge-labels. \begin{crll} $\hopfpos$ is a connected Hopf algebra. \end{crll} \begin{proof} $\hopfpos$ is graded by the value of $\nu(\hat{1})$. There is only one element of degree $0$ because $\nu$ must be strictly order preserving. It follows that $\hopfpos$ is a graded, connected bialgebra and therefore a Hopf algebra \cite{manchon2004}. \end{proof}

\subsection{A Hopf algebra morphism from Feynman diagrams to lattices}

\begin{thm} Let $\nu(\gamma) = h_1(\gamma)$. \label{thm:hopf_alg_morph} The map, \begin{align} &\chi_D:& &\hopffg_D& &\rightarrow& &\hopfpos, \\ && &\Gamma& &\mapsto& &( \sdsubdiags_D(\Gamma), \nu ), \end{align} which assigns to every diagram, its poset of s.d.\ subdiagrams decorated by the loop number of the subdiagram, is a Hopf algebra morphism. \end{thm} \begin{proof} First, it needs to be shown that $\chi_D$ is an algebra morphism: $\chi_D \circ m_{\hopffg_D} = m_{\hopfpos} \circ ( \chi_D \otimes \chi_D )$. It is sufficient to prove this for the product of two generators $\Gamma_1, \Gamma_2 \in \hopffg_D$. Subdiagrams of the product $m( \Gamma_1 \otimes \Gamma_2 ) = \Gamma_1 \Gamma_2$, which is defined as the disjoint union of the two diagrams, can be represented as pairs $(\gamma_1, \gamma_2)$ where $(\gamma_1, \gamma_2) \subset \Gamma_1 \Gamma_2$ if $\gamma_1 \subset \Gamma_1$ and $\gamma_2 \subset \Gamma_2$. This corresponds to the Cartesian product regarding the poset structure of the subdivergences. The loop number of such a pair is the sum of the loop numbers of the components. On those grounds: \begin{align} \chi_D( \Gamma_1 \Gamma_2 ) = & ( \sdsubdiags_D( \Gamma_1 \Gamma_2 ), \nu ) = ( \sdsubdiags_D( \Gamma_1 ) \times \sdsubdiags_D( \Gamma_2 ), \nu_1 + \nu_2 ) \\ =&m_{\hopfpos} ( \chi_D(\Gamma_1) \otimes \chi_D(\Gamma_2) ). \end{align} To prove that $\chi_D$ is a coalgebra morphism, we need to verify that,  \begin{align} \label{eqn:comorphism} (\chi_D \otimes \chi_D) \circ \Delta_{\hopffg_D} = \Delta_{\hopfpos} \circ \chi_D. \end{align} Choosing some generator $\Gamma$ of $\hopffg_D$ and using the definition of $\Delta_D$: \begin{align} (\chi_D \otimes \chi_D) \circ \Delta_{\hopffg_D} \Gamma = \sum \limits_{ \gamma \in \sdsubdiags_D(\Gamma) } \chi_D( \gamma) \otimes \chi_D( \Gamma / \gamma ), \end{align} the statement follows from $\chi_D(\gamma) = ( [ \hat{0}, \gamma ], \nu(\gamma) )$ and  $\chi_D( \Gamma / \gamma ) = ( [ \emptyset, \Gamma / \gamma ], \nu ) \simeq ( [ \gamma, \Gamma ], \nu - \nu(\gamma) ) $, which is a direct consequence of the definition of contractions in  eq. \eqref{eqn:contraction}. \end{proof}

\begin{crll} \label{crll:joinmeethopflat} In a join-meet-renormalizable QFT, $\im(\chi_D) \subset \hopflat \subset \hopfpos$, where $\hopflat$ is the subspace of $\hopfpos$ which is generated by all elements $(L, \nu)$, where $L$ is a lattice. In other words: In a join-meet-renormalizable QFT, $\chi_D$ maps s.d.\ diagrams and products of them to decorated lattices.  \end{crll}

\begin{expl} For any primitive diagram $\Gamma \in \text{Prim} \hopffg_D$, \ifdefined\nodraft \begin{align} \chi_D( \Gamma ) = ( \sdsubdiags_D(\Gamma), \nu ) = {\def \scale {4ex} \begin{tikzpicture}[x=\scale,y=\scale,baseline={([yshift=-.5ex]current bounding box.center)}] \begin{scope}[node distance=1] \coordinate (top) ; \coordinate [below= of top] (bottom); \draw (top) -- (bottom); \filldraw[fill=white, draw=black,circle] (top) node[fill,circle,draw,inner sep=1pt]{$L$}; \filldraw[fill=white, draw=black,circle] (bottom) node[fill,circle,draw,inner sep=1pt]{$0$}; \end{scope} \end{tikzpicture} }, \end{align} where the vertices in the Hasse diagram are decorated by the value of $\nu$ and  $L = h_1(\Gamma)$ is the loop number of the primitive diagram. 

The coproduct of $\chi_D(\Gamma)$ in $\hopfpos$ can be calculated using eq. \eqref{eqn:poset_cop}: \begin{align} \label{eqn:explcopposet} \Delta_{\hopfpos} {\def \scale {4ex} \begin{tikzpicture}[x=\scale,y=\scale,baseline={([yshift=-.5ex]current bounding box.center)}] \begin{scope}[node distance=1] \coordinate (top) ; \coordinate [below= of top] (bottom); \draw (top) -- (bottom); \filldraw[fill=white, draw=black,circle] (top) node[fill,circle,draw,inner sep=1pt]{$L$}; \filldraw[fill=white, draw=black,circle] (bottom) node[fill,circle,draw,inner sep=1pt]{$0$}; \end{scope} \end{tikzpicture} } = {\def \scale {4ex} \begin{tikzpicture}[x=\scale,y=\scale,baseline={([yshift=-.5ex]current bounding box.center)}] \begin{scope}[node distance=1] \coordinate (top) ; \coordinate [below= of top] (bottom); \draw (top) -- (bottom); \filldraw[fill=white, draw=black,circle] (top) node[fill,circle,draw,inner sep=1pt]{$L$}; \filldraw[fill=white, draw=black,circle] (bottom) node[fill,circle,draw,inner sep=1pt]{$0$}; \end{scope} \end{tikzpicture} } \otimes \mathbb{I} + \mathbb{I} \otimes {\def \scale {4ex} \begin{tikzpicture}[x=\scale,y=\scale,baseline={([yshift=-.5ex]current bounding box.center)}] \begin{scope}[node distance=1] \coordinate (top) ; \coordinate [below= of top] (bottom); \draw (top) -- (bottom); \filldraw[fill=white, draw=black,circle] (top) node[fill,circle,draw,inner sep=1pt]{$L$}; \filldraw[fill=white, draw=black,circle] (bottom) node[fill,circle,draw,inner sep=1pt]{$0$}; \end{scope} \end{tikzpicture} }. \end{align} As expected, these decorated posets are also primitive in $\hopfpos$. \else

MISSING IN DRAFT MODE

\fi \end{expl} \begin{expl} For the diagram \ifdefined\nodraft $ {\def \scale {1.5ex} \begin{tikzpicture}[x=\scale,y=\scale,baseline={([yshift=-.5ex]current bounding box.center)}] \begin{scope}[node distance=1] \coordinate (v0); \coordinate[right=.5 of v0] (v4); \coordinate[above right= of v4] (v2); \coordinate[below right= of v4] (v3); \coordinate[below right= of v2] (v5); \coordinate[right=.5 of v5] (v1); \coordinate[above right= of v2] (o1); \coordinate[below right= of v2] (o2); \coordinate[below left=.5 of v0] (i1); \coordinate[above left=.5 of v0] (i2); \coordinate[below right=.5 of v1] (o1); \coordinate[above right=.5 of v1] (o2); \draw (v0) -- (i1); \draw (v0) -- (i2); \draw (v1) -- (o1); \draw (v1) -- (o2); \draw (v0) to[bend left=20] (v2); \draw (v0) to[bend right=20] (v3); \draw (v1) to[bend left=20] (v3); \draw (v1) to[bend right=20] (v2); \draw (v2) to[bend right=60] (v3); \draw (v2) to[bend left=60] (v3); \filldraw (v0) circle(1pt); \filldraw (v1) circle(1pt); \filldraw (v2) circle(1pt); \filldraw (v3) circle(1pt); \ifdefined\cvl \draw[line width=1.5pt] (v0) to[bend left=20] (v2); \draw[line width=1.5pt] (v0) to[bend right=20] (v3); \fi \ifdefined\cvr \draw[line width=1.5pt] (v1) to[bend left=20] (v3); \draw[line width=1.5pt] (v1) to[bend right=20] (v2); \fi \ifdefined\cvml \draw[line width=1.5pt] (v2) to[bend left=60] (v3); \fi \ifdefined\cvmr \draw[line width=1.5pt] (v2) to[bend right=60] (v3); \fi \end{scope} \end{tikzpicture} } \in \hopffg_4$, $\chi_D$ gives the decorated poset, \begin{align} \chi_D\left( {\def \scale {4ex} \begin{tikzpicture}[x=\scale,y=\scale,baseline={([yshift=-.5ex]current bounding box.center)}] \begin{scope}[node distance=1] \coordinate (v0); \coordinate[right=.5 of v0] (v4); \coordinate[above right= of v4] (v2); \coordinate[below right= of v4] (v3); \coordinate[below right= of v2] (v5); \coordinate[right=.5 of v5] (v1); \coordinate[above right= of v2] (o1); \coordinate[below right= of v2] (o2); \coordinate[below left=.5 of v0] (i1); \coordinate[above left=.5 of v0] (i2); \coordinate[below right=.5 of v1] (o1); \coordinate[above right=.5 of v1] (o2); \draw (v0) -- (i1); \draw (v0) -- (i2); \draw (v1) -- (o1); \draw (v1) -- (o2); \draw (v0) to[bend left=20] (v2); \draw (v0) to[bend right=20] (v3); \draw (v1) to[bend left=20] (v3); \draw (v1) to[bend right=20] (v2); \draw (v2) to[bend right=60] (v3); \draw (v2) to[bend left=60] (v3); \filldraw (v0) circle(1pt); \filldraw (v1) circle(1pt); \filldraw (v2) circle(1pt); \filldraw (v3) circle(1pt); \ifdefined\cvl \draw[line width=1.5pt] (v0) to[bend left=20] (v2); \draw[line width=1.5pt] (v0) to[bend right=20] (v3); \fi \ifdefined\cvr \draw[line width=1.5pt] (v1) to[bend left=20] (v3); \draw[line width=1.5pt] (v1) to[bend right=20] (v2); \fi \ifdefined\cvml \draw[line width=1.5pt] (v2) to[bend left=60] (v3); \fi \ifdefined\cvmr \draw[line width=1.5pt] (v2) to[bend right=60] (v3); \fi \end{scope} \end{tikzpicture} } \right) = {\def \scale {4ex} \begin{tikzpicture}[x=\scale,y=\scale,baseline={([yshift=-.5ex]current bounding box.center)}] \begin{scope}[node distance=1] \coordinate (top) ; \coordinate [below left= of top] (v1); \coordinate [below right= of top] (v2); \coordinate [below left= of v2] (v3); \coordinate [below= of v3] (bottom); \draw (top) -- (v1); \draw (top) -- (v2); \draw (v1) -- (v3); \draw (v2) -- (v3); \draw (v3) -- (bottom); \filldraw[fill=white, draw=black,circle] (top) node[fill,circle,draw,inner sep=1pt]{$3$}; \filldraw[fill=white, draw=black,circle] (v1) node[fill,circle,draw,inner sep=1pt]{$2$}; \filldraw[fill=white, draw=black,circle] (v2) node[fill,circle,draw,inner sep=1pt]{$2$}; \filldraw[fill=white, draw=black,circle] (v3) node[fill,circle,draw,inner sep=1pt]{$1$}; \filldraw[fill=white, draw=black,circle] (bottom) node[fill,circle,draw,inner sep=1pt]{$0$}; \end{scope} \end{tikzpicture} }, \end{align} \else

MISSING IN DRAFT MODE

\fi \ifdefined\nodraft of which the reduced coproduct in $\hopfpos$ can be calculated, \begin{align} \label{eqn:explcopposet2} \widetilde{\Delta}_{\hopfpos} {\def \scale {4ex} \begin{tikzpicture}[x=\scale,y=\scale,baseline={([yshift=-.5ex]current bounding box.center)}] \begin{scope}[node distance=1] \coordinate (top) ; \coordinate [below left= of top] (v1); \coordinate [below right= of top] (v2); \coordinate [below left= of v2] (v3); \coordinate [below= of v3] (bottom); \draw (top) -- (v1); \draw (top) -- (v2); \draw (v1) -- (v3); \draw (v2) -- (v3); \draw (v3) -- (bottom); \filldraw[fill=white, draw=black,circle] (top) node[fill,circle,draw,inner sep=1pt]{$3$}; \filldraw[fill=white, draw=black,circle] (v1) node[fill,circle,draw,inner sep=1pt]{$2$}; \filldraw[fill=white, draw=black,circle] (v2) node[fill,circle,draw,inner sep=1pt]{$2$}; \filldraw[fill=white, draw=black,circle] (v3) node[fill,circle,draw,inner sep=1pt]{$1$}; \filldraw[fill=white, draw=black,circle] (bottom) node[fill,circle,draw,inner sep=1pt]{$0$}; \end{scope} \end{tikzpicture} } = 2~ {\def \scale {4ex} \begin{tikzpicture}[x=\scale,y=\scale,baseline={([yshift=-.5ex]current bounding box.center)}] \begin{scope}[node distance=1] \coordinate (top) ; \coordinate [below= of top] (v1); \coordinate [below= of v1] (bottom); \draw (top) -- (v1); \draw (v1) -- (bottom); \filldraw[fill=white, draw=black,circle] (top) node[fill,circle,draw,inner sep=1pt]{$2$}; \filldraw[fill=white, draw=black,circle] (v1) node[fill,circle,draw,inner sep=1pt]{$1$}; \filldraw[fill=white, draw=black,circle] (bottom) node[fill,circle,draw,inner sep=1pt]{$0$}; \end{scope} \end{tikzpicture} } \otimes { \def \scale {4ex} \begin{tikzpicture}[x=\scale,y=\scale,baseline={([yshift=-.5ex]current bounding box.center)}] \begin{scope}[node distance=1] \coordinate (top) ; \coordinate [below= of top] (bottom); \draw (top) -- (bottom); \filldraw[fill=white, draw=black,circle] (top) node[fill,circle,draw,inner sep=1pt]{$1$}; \filldraw[fill=white, draw=black,circle] (bottom) node[fill,circle,draw,inner sep=1pt]{$0$}; \end{scope} \end{tikzpicture} } + { \def \scale {4ex} \begin{tikzpicture}[x=\scale,y=\scale,baseline={([yshift=-.5ex]current bounding box.center)}] \begin{scope}[node distance=1] \coordinate (top) ; \coordinate [below= of top] (bottom); \draw (top) -- (bottom); \filldraw[fill=white, draw=black,circle] (top) node[fill,circle,draw,inner sep=1pt]{$1$}; \filldraw[fill=white, draw=black,circle] (bottom) node[fill,circle,draw,inner sep=1pt]{$0$}; \end{scope} \end{tikzpicture} } \otimes { \def \scale {4ex} \begin{tikzpicture}[x=\scale,y=\scale,baseline={([yshift=-.5ex]current bounding box.center)}] \begin{scope}[node distance=1] \coordinate (top) ; \coordinate [below left= of top] (v1); \coordinate [below right= of top] (v2); \coordinate [below left= of v2] (bottom); \draw (top) -- (v1); \draw (top) -- (v2); \draw (v1) -- (bottom); \draw (v2) -- (bottom); \filldraw[fill=white, draw=black,circle] (top) node[fill,circle,draw,inner sep=1pt]{$2$}; \filldraw[fill=white, draw=black,circle] (v1) node[fill,circle,draw,inner sep=1pt]{$1$}; \filldraw[fill=white, draw=black,circle] (v2) node[fill,circle,draw,inner sep=1pt]{$1$}; \filldraw[fill=white, draw=black,circle] (bottom) node[fill,circle,draw,inner sep=1pt]{$0$}; \end{scope} \end{tikzpicture} }. \end{align} This can be compared to the coproduct calculation in example \ref{expl:coproducthopffg},  \begin{align} \widetilde \Delta_4 { \def \scale {2ex} \begin{tikzpicture}[x=\scale,y=\scale,baseline={([yshift=-.5ex]current bounding box.center)}] \begin{scope}[node distance=1] \coordinate (v0); \coordinate[right=.5 of v0] (v4); \coordinate[above right= of v4] (v2); \coordinate[below right= of v4] (v3); \coordinate[below right= of v2] (v5); \coordinate[right=.5 of v5] (v1); \coordinate[above right= of v2] (o1); \coordinate[below right= of v2] (o2); \coordinate[below left=.5 of v0] (i1); \coordinate[above left=.5 of v0] (i2); \coordinate[below right=.5 of v1] (o1); \coordinate[above right=.5 of v1] (o2); \draw (v0) -- (i1); \draw (v0) -- (i2); \draw (v1) -- (o1); \draw (v1) -- (o2); \draw (v0) to[bend left=20] (v2); \draw (v0) to[bend right=20] (v3); \draw (v1) to[bend left=20] (v3); \draw (v1) to[bend right=20] (v2); \draw (v2) to[bend right=60] (v3); \draw (v2) to[bend left=60] (v3); \filldraw (v0) circle(1pt); \filldraw (v1) circle(1pt); \filldraw (v2) circle(1pt); \filldraw (v3) circle(1pt); \ifdefined\cvl \draw[line width=1.5pt] (v0) to[bend left=20] (v2); \draw[line width=1.5pt] (v0) to[bend right=20] (v3); \fi \ifdefined\cvr \draw[line width=1.5pt] (v1) to[bend left=20] (v3); \draw[line width=1.5pt] (v1) to[bend right=20] (v2); \fi \ifdefined\cvml \draw[line width=1.5pt] (v2) to[bend left=60] (v3); \fi \ifdefined\cvmr \draw[line width=1.5pt] (v2) to[bend right=60] (v3); \fi \end{scope} \end{tikzpicture} } &= 2 { \def \scale {2ex} \begin{tikzpicture}[x=\scale,y=\scale,baseline={([yshift=-.5ex]current bounding box.center)}] \begin{scope}[node distance=1] \coordinate (v0); \coordinate[right=.5 of v0] (v4); \coordinate[above right= of v4] (v2); \coordinate[below right= of v4] (v3); \coordinate[above right=.5 of v2] (o1); \coordinate[below right=.5 of v3] (o2); \coordinate[below left=.5 of v0] (i1); \coordinate[above left=.5 of v0] (i2); \draw (v0) -- (i1); \draw (v0) -- (i2); \draw (v2) -- (o1); \draw (v3) -- (o2); \draw (v0) to[bend left=20] (v2); \draw (v0) to[bend right=20] (v3); \draw (v2) to[bend right=60] (v3); \draw (v2) to[bend left=60] (v3); \filldraw (v0) circle(1pt); \filldraw (v2) circle(1pt); \filldraw (v3) circle(1pt); \end{scope} \end{tikzpicture} } \otimes { \def \scale {2ex} \begin{tikzpicture}[x=\scale,y=\scale,baseline={([yshift=-.5ex]current bounding box.center)}] \begin{scope}[node distance=1] \coordinate (v0); \coordinate [right=.5 of v0] (vm); \coordinate [right=.5 of vm] (v1); \coordinate [above left=.5 of v0] (i0); \coordinate [below left=.5 of v0] (i1); \coordinate [above right=.5 of v1] (o0); \coordinate [below right=.5 of v1] (o1); \draw (vm) circle(.5); \draw (i0) -- (v0); \draw (i1) -- (v0); \draw (o0) -- (v1); \draw (o1) -- (v1); \filldraw (v0) circle(1pt); \filldraw (v1) circle(1pt); \end{scope} \end{tikzpicture} } + { \def \scale {2ex} \begin{tikzpicture}[x=\scale,y=\scale,baseline={([yshift=-.5ex]current bounding box.center)}] \begin{scope}[node distance=1] \coordinate (v0); \coordinate [below=1 of v0] (v1); \coordinate [above left=.5 of v0] (i0); \coordinate [above right=.5 of v0] (i1); \coordinate [below left=.5 of v1] (o0); \coordinate [below right=.5 of v1] (o1); \coordinate [above=.5 of v1] (vm); \draw (vm) circle(.5); \draw (i0) -- (v0); \draw (i1) -- (v0); \draw (o0) -- (v1); \draw (o1) -- (v1); \filldraw (v0) circle(1pt); \filldraw (v1) circle(1pt); \end{scope} \end{tikzpicture} } \otimes { \def \scale {2ex} \begin{tikzpicture}[x=\scale,y=\scale,baseline={([yshift=-.5ex]current bounding box.center)}] \begin{scope}[node distance=1] \coordinate (v0); \coordinate [right=.5 of v0] (vm1); \coordinate [right=.5 of vm1] (v1); \coordinate [right=.5 of v1] (vm2); \coordinate [right=.5 of vm2] (v2); \coordinate [above left=.5 of v0] (i0); \coordinate [below left=.5 of v0] (i1); \coordinate [above right=.5 of v2] (o0); \coordinate [below right=.5 of v2] (o1); \draw (vm1) circle(.5); \draw (vm2) circle(.5); \draw (i0) -- (v0); \draw (i1) -- (v0); \draw (o0) -- (v2); \draw (o1) -- (v2); \filldraw (v0) circle(1pt); \filldraw (v1) circle(1pt); \filldraw (v2) circle(1pt); \end{scope} \end{tikzpicture} } \end{align} and identity eq. \eqref{eqn:comorphism} is verified after computing the decorated poset of each subdiagram of  $ {\def \scale {1.5ex} \begin{tikzpicture}[x=\scale,y=\scale,baseline={([yshift=-.5ex]current bounding box.center)}] \begin{scope}[node distance=1] \coordinate (v0); \coordinate[right=.5 of v0] (v4); \coordinate[above right= of v4] (v2); \coordinate[below right= of v4] (v3); \coordinate[below right= of v2] (v5); \coordinate[right=.5 of v5] (v1); \coordinate[above right= of v2] (o1); \coordinate[below right= of v2] (o2); \coordinate[below left=.5 of v0] (i1); \coordinate[above left=.5 of v0] (i2); \coordinate[below right=.5 of v1] (o1); \coordinate[above right=.5 of v1] (o2); \draw (v0) -- (i1); \draw (v0) -- (i2); \draw (v1) -- (o1); \draw (v1) -- (o2); \draw (v0) to[bend left=20] (v2); \draw (v0) to[bend right=20] (v3); \draw (v1) to[bend left=20] (v3); \draw (v1) to[bend right=20] (v2); \draw (v2) to[bend right=60] (v3); \draw (v2) to[bend left=60] (v3); \filldraw (v0) circle(1pt); \filldraw (v1) circle(1pt); \filldraw (v2) circle(1pt); \filldraw (v3) circle(1pt); \ifdefined\cvl \draw[line width=1.5pt] (v0) to[bend left=20] (v2); \draw[line width=1.5pt] (v0) to[bend right=20] (v3); \fi \ifdefined\cvr \draw[line width=1.5pt] (v1) to[bend left=20] (v3); \draw[line width=1.5pt] (v1) to[bend right=20] (v2); \fi \ifdefined\cvml \draw[line width=1.5pt] (v2) to[bend left=60] (v3); \fi \ifdefined\cvmr \draw[line width=1.5pt] (v2) to[bend right=60] (v3); \fi \end{scope} \end{tikzpicture} }$ and comparing the previous two equations: \begin{align} &\chi_4 \left( {\def \scale {4ex} \begin{tikzpicture}[x=\scale,y=\scale,baseline={([yshift=-.5ex]current bounding box.center)}] \begin{scope}[node distance=1] \coordinate (v0); \coordinate[right=.5 of v0] (v4); \coordinate[above right= of v4] (v2); \coordinate[below right= of v4] (v3); \coordinate[above right=.5 of v2] (o1); \coordinate[below right=.5 of v3] (o2); \coordinate[below left=.5 of v0] (i1); \coordinate[above left=.5 of v0] (i2); \draw (v0) -- (i1); \draw (v0) -- (i2); \draw (v2) -- (o1); \draw (v3) -- (o2); \draw (v0) to[bend left=20] (v2); \draw (v0) to[bend right=20] (v3); \draw (v2) to[bend right=60] (v3); \draw (v2) to[bend left=60] (v3); \filldraw (v0) circle(1pt); \filldraw (v2) circle(1pt); \filldraw (v3) circle(1pt); \end{scope} \end{tikzpicture} } \right) = {\def \scale {4ex} \begin{tikzpicture}[x=\scale,y=\scale,baseline={([yshift=-.5ex]current bounding box.center)}] \begin{scope}[node distance=1] \coordinate (top) ; \coordinate [below= of top] (v1); \coordinate [below= of v1] (bottom); \draw (top) -- (v1); \draw (v1) -- (bottom); \filldraw[fill=white, draw=black,circle] (top) node[fill,circle,draw,inner sep=1pt]{$2$}; \filldraw[fill=white, draw=black,circle] (v1) node[fill,circle,draw,inner sep=1pt]{$1$}; \filldraw[fill=white, draw=black,circle] (bottom) node[fill,circle,draw,inner sep=1pt]{$0$}; \end{scope} \end{tikzpicture} }& &\chi_4 \left( { \def \scale {4ex} \begin{tikzpicture}[x=\scale,y=\scale,baseline={([yshift=-.5ex]current bounding box.center)}] \begin{scope}[node distance=1] \coordinate (v0); \coordinate [right=.5 of v0] (vm); \coordinate [right=.5 of vm] (v1); \coordinate [above left=.5 of v0] (i0); \coordinate [below left=.5 of v0] (i1); \coordinate [above right=.5 of v1] (o0); \coordinate [below right=.5 of v1] (o1); \draw (vm) circle(.5); \draw (i0) -- (v0); \draw (i1) -- (v0); \draw (o0) -- (v1); \draw (o1) -- (v1); \filldraw (v0) circle(1pt); \filldraw (v1) circle(1pt); \end{scope} \end{tikzpicture} } \right) = { \def \scale {4ex} \begin{tikzpicture}[x=\scale,y=\scale,baseline={([yshift=-.5ex]current bounding box.center)}] \begin{scope}[node distance=1] \coordinate (top) ; \coordinate [below= of top] (bottom); \draw (top) -- (bottom); \filldraw[fill=white, draw=black,circle] (top) node[fill,circle,draw,inner sep=1pt]{$1$}; \filldraw[fill=white, draw=black,circle] (bottom) node[fill,circle,draw,inner sep=1pt]{$0$}; \end{scope} \end{tikzpicture} } & &\chi_4 \left( {\def \scale {4ex} \begin{tikzpicture}[x=\scale,y=\scale,baseline={([yshift=-.5ex]current bounding box.center)}] \begin{scope}[node distance=1] \coordinate (v0); \coordinate [right=.5 of v0] (vm1); \coordinate [right=.5 of vm1] (v1); \coordinate [right=.5 of v1] (vm2); \coordinate [right=.5 of vm2] (v2); \coordinate [above left=.5 of v0] (i0); \coordinate [below left=.5 of v0] (i1); \coordinate [above right=.5 of v2] (o0); \coordinate [below right=.5 of v2] (o1); \draw (vm1) circle(.5); \draw (vm2) circle(.5); \draw (i0) -- (v0); \draw (i1) -- (v0); \draw (o0) -- (v2); \draw (o1) -- (v2); \filldraw (v0) circle(1pt); \filldraw (v1) circle(1pt); \filldraw (v2) circle(1pt); \end{scope} \end{tikzpicture} } \right) = {\def \scale {4ex} \begin{tikzpicture}[x=\scale,y=\scale,baseline={([yshift=-.5ex]current bounding box.center)}] \begin{scope}[node distance=1] \coordinate (top) ; \coordinate [below left= of top] (v1); \coordinate [below right= of top] (v2); \coordinate [below left= of v2] (bottom); \draw (top) -- (v1); \draw (top) -- (v2); \draw (v1) -- (bottom); \draw (v2) -- (bottom); \filldraw[fill=white, draw=black,circle] (top) node[fill,circle,draw,inner sep=1pt]{$2$}; \filldraw[fill=white, draw=black,circle] (v1) node[fill,circle,draw,inner sep=1pt]{$1$}; \filldraw[fill=white, draw=black,circle] (v2) node[fill,circle,draw,inner sep=1pt]{$1$}; \filldraw[fill=white, draw=black,circle] (bottom) node[fill,circle,draw,inner sep=1pt]{$0$}; \end{scope} \end{tikzpicture} }. \end{align} \else

MISSING IN DRAFT MODE

\fi \end{expl} \section{Properties of the lattices of subdivergences} \label{sec:properties}

Although, the above Hopf algebra morphism can be applied in every renormalizable QFT, we shall restrict ourselves to join-meet-renormalizable QFTs, where $\chi_D$ maps to $\hopflat$, the Hopf algebra of decorated lattices, as a result of corollary \ref{crll:joinmeethopflat}. 

The decorated lattice, which is associated to a Feynman diagram, encodes the `overlappingness' of the diagrams' subdivergences. Different join-meet-renormalizable QFTs have quite distinguished properties in this respect. Interestingly, the types of the decorated lattices appearing depend on the residues or equivalently on the superficial degree of divergence of the diagrams under consideration. For instance, it was proven by Berghoff in the context of Wonderful models that every diagram with only logarithmically divergent subdivergences (i.e. $\forall \gamma \in \sdsubdiags_D(\Gamma): \omega_D(\gamma) = 0$) is \textit{distributive}: \begin{prop}{\cite[Prop. 3.22]{berghoff2014wonderful}} If $\Gamma$ has only logarithmically s.d.\ subdiagrams in $D$ dimensions, (i.e. $\forall \gamma \in \sdsubdiags_D(\Gamma) \Rightarrow \omega_D(\gamma) = 0$), then the distributivity identities, \begin{align} \gamma_1 \meet ( \gamma_2 \join \gamma_3) &= ( \gamma_1 \meet \gamma_2 ) \join ( \gamma_1 \meet \gamma_3 ) \\ \gamma_1 \join ( \gamma_2 \meet \gamma_3) &= ( \gamma_1 \join \gamma_2 ) \meet ( \gamma_1 \join \gamma_3 ), \end{align} hold for $\gamma_1, \gamma_2, \gamma_3 \in \sdsubdiags_D(\Gamma)$. \end{prop}

This is a pleasant result for diagrams with only logarithmically divergent subdiagrams. Because distributive lattices are always graded, this implies that we have a bigrading on $\hopflat$ for these elements. One grading by the value of $\nu(\hat{1})$, corresponding to the loop number of the diagram, and one grading by the length of the maximal chains of the lattice, which coincides with the \textit{coradical degree} of the diagram in $\hopffg_D$. The coradical filtration of $\hopffg_D$, defined in eq. \eqref{eqn:coradfilt}, consequently becomes a grading for the subspaces generated by only logarithmically s.d.\ diagrams.

\subsection{Theories with only three-or-less-valent vertices} From the preceding result the question arises what structure is left, if we also allow subdiagrams which are not only logarithmically divergent. In renormalizable QFTs with only three-or-less-valent vertices, the lattices $\sdsubdiags_D(\Gamma)$ will turn out to be semimodular. This is a weaker property than distributivity, but it still guarantees that the lattices are graded. To capture this property of $\sdsubdiags_D(\Gamma)$, some additional terms of lattice theory will be repeated following \cite{stanley1997}.

\paragraph{Join-irreducible element} An element $x$ of a lattice $L$, $x\in L$ is called join-irreducible if $x = y \join z$ always implies $x=y$ or $x=z$.

\paragraph{Atoms and coatoms} An element $x$ of $L$ is an atom of $L$ if it covers $\hat{0}$. It is a coatom of $L$ if $\hat{1}$ covers $x$.

\paragraph{Semimodular lattice} A lattice $L$ is semimodular if for two elements $x,y\in L$ that cover $x \meet y$, $x$ and $y$ are covered by $x \join y$.

With these notions we can formulate \begin{lmm} \label{lmm:propagatorunion} If in a renormalizable QFT with only three-or-less-valent vertices $\gamma_1$ and $\gamma_2$ are overlapping, they must be of vertex-type and $\gamma_1\cup \gamma_2$ of propagator-type. \end{lmm} \begin{proof} From lemma \ref{lmm:sgidentities} and proposition \ref{prop:union_4ext} we know $|\Hext(\gamma_1)| + |\Hext(\gamma_2)| = |\Hext(\gamma_1 \cup \gamma_2)| + |\Hext(\gamma_1 \cap \gamma_2)|$ and $|\Hext(\gamma_1 \cap \gamma_2)| \geq 4$. The statement follows immediately, because in a renormalizable QFT with three-or-less-valent vertices every subdivergence either has two or three external legs.  \end{proof}

\begin{crll} \label{crll:joinirreducible} In a QFT with only three-or-less-valent vertices, vertex-type s.d.\ diagrams $\Gamma$ ($|\Hext(\Gamma)|=3$) are always join-irreducible elements of $\sdsubdiags_D(\Gamma)$. \end{crll} \begin{proof} Suppose there were $\gamma_1, \gamma_2 \in \sdsubdiags_D(\Gamma)$ with $\gamma_1\neq \Gamma$, $\gamma_2 \neq \Gamma$ and $\gamma_1 \join \gamma_2 = \Gamma$. The subdivergences $\gamma_1$ and $\gamma_2$ are therefore overlapping. As lemma \ref{lmm:propagatorunion} requires $\Gamma$ to be of propagator type, we have a contradiction. \end{proof} 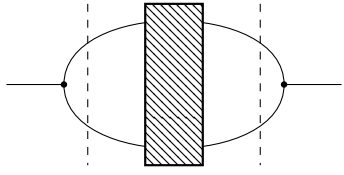
\begin{figure} {
\ifdefined\nodraft \begin{center} ${ \def\scale {10ex} \begin{tikzpicture}[x=\scale,y=\scale,baseline={([yshift=-.5ex]current bounding box.center)}] \begin{scope}[node distance=1] \coordinate (i); \coordinate[right=.5 of i] (vl); \coordinate[above right= of vl] (v1); \coordinate[below right= of vl] (v2); \coordinate[right=.5 of v1] (v3); \coordinate[right=.5 of v2] (v4); \coordinate[below right= of v3] (vr); \coordinate[right=.5 of vr] (o); \draw (i) -- (vl); \draw (o) -- (vr); \draw (vl) to[bend right=90] (vr); \draw (vl) to[bend left=90] (vr); \draw [dashed] ($(v1) - (.5,0)$) -- ($(v2) - (.5,0)$); \draw [dashed] ($(v3) + (.5,0)$) -- ($(v4) + (.5,0)$); \filldraw (vl) circle(1pt); \filldraw (vr) circle(1pt); \draw [fill=white] (v1) rectangle (v4); \draw [fill=white,thick,pattern=north west lines] (v1) rectangle (v4); \end{scope} \end{tikzpicture}}$ \end{center} \else

MISSING IN DRAFT MODE

\fi }
\caption{Structure of overlapping divergences in three-valent QFTs} \label{fig:blob33} \end{figure}

\begin{prop} \label{prop:semimodular} In a renormalizable QFT with only three-or-less-valent vertices, the lattice $\sdsubdiags_D(\Gamma)$ is semimodular for every Feynman diagram $\Gamma$.  \end{prop} \begin{proof} For two diagrams $\mu_1, \mu_2 \in \sdsubdiags_D(\Gamma)$ we can always form the contractions by $\mu_1 \meet \mu_2$: $\mu_1 / (\mu_1 \meet \mu_2)$ and $\mu_2 / (\mu_1 \meet \mu_2)$. Hence, the statement that $\mu_1,\mu_2$ cover $\mu_1 \meet \mu_2$ is equivalent to stating that $\mu_1 / (\mu_1 \meet \mu_2)$ and $\mu_2 / (\mu_1 \meet \mu_2)$ are primitive.

To prove that $\mu_1 \join \mu_2$ covers $\mu_1$ and $\mu_2$ if $\mu_1$ and $\mu_2$ cover $\mu_1 \meet \mu_2$, it is therefore sufficient to verify that for $\gamma_1,\gamma_2$ primitive $(\gamma_1\cup \gamma_2) / \gamma_1$ and $(\gamma_1\cup \gamma_2) / \gamma_2$ are primitive as well. This is obvious if $\gamma_1, \gamma_2$ are not overlapping. 

If $\gamma_1$ and $\gamma_2$ are overlapping, they must be of vertex-type and $\gamma_1\cup \gamma_2$ of propagator-type as proven in lemma \ref{lmm:propagatorunion}. Because only three-valent vertices are allowed and $|\Hext(\gamma_i) \cap \Hint(\gamma_1\cup \gamma_2)|\geq 2$ (corollary \ref{crll:1PIsdHext}), each $\gamma_1$ and $\gamma_2$ must provide one external edge for $\gamma_1 \cup \gamma_2$. The situation is depicted in fig. \ref{fig:blob33}. For both $\gamma_1$ and $\gamma_2$ to be primitive, they must share the same four-leg kernel, depicted as a striped box. Contraction with either $\gamma_1$ or $\gamma_2$ results in a one-loop propagator, which is primitive. \end{proof}

Semimodular lattices have a very rich structure, see for instance Stern's book \cite{stern1999semimodular}. Eventually, semimodularity implies that the lattices under consideration are graded: \begin{thm} \label{thm:gradthree} In a renormalizable QFT with only three-or-less-valent vertices: \begin{itemize} \item $\sdsubdiags_D(\Gamma)$ is a graded lattice for every propagator, vertex-type diagram or disjoint unions of both. \item $\hopflat$ is bigraded by $\nu(\hat{1})$ and the length of the maximal chains of the lattices, which coincides with the coradical degree in $\hopflat$. \item $\hopffg_D$ is bigraded by $h_1(\Gamma)$ and the coradical degree of $\Gamma$.  \item Every complete forest of $\Gamma$ has the same length. \end{itemize} \end{thm} \begin{proof} Every semimodular lattice is graded \cite[Proposition 3.3.2]{stanley1997}. \end{proof}

\subsection{Theories with only four-or-less-valent vertices}

{ \ifdefined\nodraft \begin{figure}   \subcaptionbox{Example of a diagram where $\sdsubdiags_4(\Gamma)$ forms a non-graded lattice.\label{fig:diagramnonranked}}   [.45\linewidth]{       $\Gamma = { \def \scale{2em} 
 }$   }   \subcaptionbox{The Hasse diagram of the corresponding non-graded lattice, where the decoration was omitted.\label{fig:diagramnonrankedlattice}}   [.45\linewidth]{   $\chi_4(\Gamma) = { \def \scale{1em} %
 }$   }   \subcaptionbox{The non-trivial superficially divergent subdiagrams and the complete forests which can be formed out of them.\label{fig:diagramnonrankeddetails}}   [\linewidth]{ \def \scale{1.5em} $ \begin{aligned} \alpha_1 &= { \def\cvlm {} \def\cvlb {} \def\cvmt {} \def\cvmm {} \def\cvmb {} \def\cvrt {} \def\cvrm {} \def\cvrb {} %
 } & & &\end{aligned}$       \\ with the complete forests $\emptyset \subset \delta_1 \subset \alpha_i \subset \Gamma$,  $\emptyset \subset \delta_2 \subset \beta_i \subset \Gamma$ and $\emptyset \subset \gamma_i \subset \Gamma$.   } \caption{Counter example of a lattice, which appears in join-meet-renormalizable QFTs with four-valent vertices and is not graded.} \label{fig:counterexplsemimod} \end{figure} \else

MISSING IN DRAFT MODE

\fi }

We have shown that every lattice associated to a s.d.\ diagram in a QFT with only three-or-less-valent vertices is semimodular. For join-meet-renormalizable QFTs which also have four-valent vertices the situation is more involved as the example in fig. \ref{fig:counterexplsemimod} exposes. The depicted lattice in fig. \ref{fig:diagramnonrankedlattice} associated to the $\phi^4$-diagram in fig. \ref{fig:diagramnonranked} is not semimodular, because it is not graded. This implies that not all complete forests are of the same length in theories, where this topology can appear. This includes $\phi^4$ and Yang-Mills theory in four dimensions. The s.d.\ subdiagrams of the counter example are illustrated in fig. \ref{fig:diagramnonrankeddetails}. It can be seen that there are six complete forests of length four and three complete forests of length three.

The pleasant property of semimodularity can be recovered by working in the Hopf algebra of Feynman diagrams without tadpoles or equivalently by setting all tadpole diagrams to zero. This is quite surprising, because the independence of loops in tadpoles from external momenta and the combinatorial structure of BPHZ, encoded by the Hopf algebra of Feynman diagrams, seem independent on the first sight. 

Formally, we can transfer the restriction to tadpole-free diagrams to $\hopflat$ by the following procedure: The Hopf algebra morphism $\psi: \hopffg_D \rightarrow \hopffgs_D$ defined in eq. \eqref{eqn:morphismpsi} gives rise to the ideal $\ker \psi \subset \hopffg_D$. Using the Hopf algebra morphism $\chi_D$ an ideal of $\hopflat$, $\chi_D(\ker \psi ) \subset \hopflat$, is obtained. This can be summarized in a commutative diagram:

\begin{center} \begin{tikzpicture} \node (tl) {$\hopffg$}; \node [right=of tl] (tr) {$\hopffgs$}; \node [below =of tl] (bl) {$\hopflat$}; \node [below =of tr] (br) {$\hopflats$}; \draw[->] (tl) -- node[above] {$\psi$} (tr); \draw[->] (tl) to node[auto] {$\chi_D$} (bl); \draw[->] (bl) -- node[auto] {$\psi'$} (br); \draw[->] (tr) -- node[right] {$\chi_D'$} (br); \end{tikzpicture} \end{center} where $\hopflats$ is the quotient $\hopflats:= \hopflat/\chi_D(\ker\psi)$ and $\psi'$ is just the projection to $\hopflats$. 

The interesting part is the morphism $\chi_D': \hopffgs_D \rightarrow \hopflats$, which maps from the Hopf algebra of Feynman diagrams without tadpoles to $\hopflats$. Such a map can be constructed explicitly and for theories with only four-or-less-valent vertices, it can be ensured that $\chi_D'$ maps Feynman diagrams to decorated semimodular lattices. 

\begin{prop} In a renormalizable QFT with only four-or-less-valent vertices, $\chi_D'$ maps elements from the Hopf algebra of Feynman diagrams without tadpoles to decorated lattices. \end{prop} \begin{proof} Explicitly, $\chi_D'$ is the map, \begin{align} \chi_D': \Gamma \mapsto (\sdsubdiagsn_D (\Gamma), \nu), \end{align} where the decoration $\nu$ is the same as above.

We need to show that $\sdsubdiagsn_D(\Gamma)$ ordered by inclusion is a lattice. This is not as simple as before, because $\gamma_1, \gamma_2 \in \sdsubdiagsn_D(\Gamma)$ does not necessarily imply $\gamma_1 \cup \gamma_2 \in \sdsubdiagsn_D(\Gamma)$. From definition \ref{def:snailfreehopf} of $\sdsubdiagsn_D(\Gamma)$, we can deduce that if $\gamma_1, \gamma_2 \in \sdsubdiagsn_D(\Gamma)$, then $\gamma_1 \cup \gamma_2 \notin \sdsubdiagsn_D(\Gamma)$ iff $\Gamma / \gamma_1 \cup \gamma_2$ is a tadpole. 

To prove that there still exists a least upper bound for every pair $\gamma_1$, $\gamma_2$ we must ensure that every element $\mu \in \sdsubdiags_D(\Gamma)$ and $\mu \notin \sdsubdiagsn_D(\Gamma)$ is only covered by one element in $\sdsubdiagsn_D(\Gamma)$. This is equivalent to stating that if $\gamma \subset \delta \subset \Gamma$ and $\delta / \gamma$ is a primitive tadpole (i.e. a self-loop with one vertex), then there is no $\delta' \subset \Gamma$ such that $\delta' / \gamma$ is a primitive tadpole. There cannot be such a second subdiagram $\delta'$. Suppose there were such $\delta$ and $\delta'$. $\delta$ and $\delta'$ are obtained from $\gamma$ by joining two of its external legs to an new edge. As only four-or-less-valent vertices are allowed, such a configuration can only be achieved if $\gamma$ is a diagram with four external legs. $\delta$ and $\delta'$ are the diagrams obtained by closing either pair of legs of $\gamma$. This would imply that $\delta_1 \cup \delta_2$ is a vacuum diagram without external legs, which is excluded. \end{proof}

\begin{expl} The map $\chi_D'$ can be applied to the example in fig. \ref{fig:counterexplsemimod} where the lattice obtained by $\chi_D$ is not semimodular. It can be seen from fig. \ref{fig:diagramnonrankeddetails} that the only diagrams, which do not result in tadpole diagrams upon contraction are $\delta_1$ and $\delta_2$. Accordingly,  \begin{align} \chi_4'\left({ \def \scale{1.5em} \begin{tikzpicture}[x=\scale,y=\scale,baseline={([yshift=-.5ex]current bounding box.center)}] \begin{scope}[node distance=1] \coordinate (i1); \coordinate[right=.5 of i1] (v0); \coordinate[right=.5 of v0] (v2); \coordinate[right=.5 of v2] (vm); \coordinate[above=of vm] (v3); \coordinate[below=of vm] (v4); \coordinate[right=.5 of vm] (v2d); \coordinate[right=.5 of v2d] (v1); \coordinate[right=.5 of v1] (o1); \ifdefined \cvlt \draw[line width=1.5pt] (v0) to[bend left=45] (v3); \else \draw (v0) to[bend left=45] (v3); \fi \ifdefined \cvlb \draw[line width=1.5pt] (v0) to[bend right=45] (v4); \else \draw (v0) to[bend right=45] (v4); \fi \ifdefined \cvlm \draw[line width=1.5pt] (v0) -- (v2); \else \draw (v0) -- (v2); \fi \ifdefined \cvrt \draw[line width=1.5pt] (v1) to[bend right=45] (v3); \else \draw (v1) to[bend right=45] (v3); \fi \ifdefined \cvrb \draw[line width=1.5pt] (v1) to[bend left=45] (v4); \else \draw (v1) to[bend left=45] (v4); \fi \ifdefined \cvmt \draw[line width=1.5pt] (v3) to[bend right=20] (v2); \else \draw (v3) to[bend right=20] (v2); \fi \ifdefined \cvmb \draw[line width=1.5pt] (v4) to[bend left=20] (v2); \else \draw (v4) to[bend left=20] (v2); \fi \ifdefined \cvmm \draw[line width=1.5pt] (v3) to[bend left=20] (v2d); \draw[line width=1.5pt] (v4) to[bend right=20] (v2d); \else \draw (v3) to[bend left=20] (v2d); \draw (v4) to[bend right=20] (v2d); \fi \filldraw[color=white] (v2d) circle(.2); \ifdefined \cvrm \draw[line width=1.5pt] (v1) -- (v2); \else \draw (v1) -- (v2); \fi \draw (v0) -- (i1); \draw (v1) -- (o1); \filldraw (v0) circle(1pt); \filldraw (v1) circle(1pt); \filldraw (v2) circle(1pt); \filldraw (v3) circle(1pt); \filldraw (v4) circle(1pt); \end{scope} \end{tikzpicture} } \right) = { \def \scale{1.5em} \begin{tikzpicture}[x=\scale,y=\scale,baseline={([yshift=-.5ex]current bounding box.center)}] \begin{scope}[node distance=1] \coordinate (top) ; \coordinate [below left= of top] (v1); \coordinate [below right= of top] (v2); \coordinate [below left= of v2] (bottom); \draw (top) -- (v1); \draw (top) -- (v2); \draw (v1) -- (bottom); \draw (v2) -- (bottom); \filldraw[fill=white, draw=black,circle] (top) node[fill,circle,draw,inner sep=1pt]{$5$}; \filldraw[fill=white, draw=black,circle] (v1) node[fill,circle,draw,inner sep=1pt]{$3$}; \filldraw[fill=white, draw=black,circle] (v2) node[fill,circle,draw,inner sep=1pt]{$3$}; \filldraw[fill=white, draw=black,circle] (bottom) node[fill,circle,draw,inner sep=1pt]{$0$}; \end{scope} \end{tikzpicture} }, \end{align} which is a graded lattice. \end{expl}

\begin{prop} In a QFT with only four-or-less-valent vertices  $\chi_D'$ maps elements from the Hopf algebra of Feynman diagrams without tadpoles to decorated semimodular lattices. \end{prop} \begin{proof} As above we only need to prove that if $\gamma_1$ and $\gamma_2$ are overlapping and primitive, then $(\gamma_1\join \gamma_2) / \gamma_1$ and $(\gamma_1\join \gamma_2) / \gamma_2$ are primitive as well. The key is to notice that $\Hext(\gamma_1 \join \gamma_2) \not \subset \Hext(\gamma_1)$ and $\Hext(\gamma_1 \join \gamma_2) \not \subset \Hext(\gamma_2)$. Otherwise, one of the contracted diagrams would be a tadpole. For this reason, we can characterize the overlapping divergences of a diagram by the proper subset of external half-edges of the full diagram it contains.  If $\gamma_1 \join \gamma_2/\gamma_1$ was not primitive, we could remove the half-edges $\Hext(\gamma_1) \cap \Hext(\gamma_1 \join \gamma_2)$ and the adjacent vertices and edges. The result would be a s.d.\ subdiagram of $\gamma_2$ in contradiction with the requirement. \end{proof}

It is interesting how important taking the quotient by the tadpole diagrams is, to obtain the property of semimodularity for the lattices of Feynman diagrams. \begin{thm} \label{thm:gradfour} In a renormalizable QFT with only four-or-less-valent vertices: \begin{itemize} \item $\sdsubdiagsn_D(\Gamma)$ is a graded lattice for every propagator, vertex-type diagram or disjoint unions of both. \item $\hopflats$ is bigraded by $\nu(\hat{1})$ and the length of the maximal chains of the lattices, which coincides with the coradical degree in $\hopflat$. \item $\hopffgs_D$ is bigraded by $h_1(\Gamma)$ and the coradical degree of $\Gamma$.  \item Every complete forest of $\Gamma$, which does not result in a tadpole upon contraction, has the same length. \end{itemize} \end{thm} \begin{proof} Every semimodular lattice is graded \cite[Proposition 3.3.2]{stanley1997}. \end{proof} The overlapping diagrams in $\sdsubdiagsn_D(\Gamma)$ are characterized by the external legs of $\Gamma$ they contain. As a consequence, there is a limited number of possibilities for primitive diagrams to be overlapping. A two-leg diagram can only be the join of at most two primitive overlapping diagrams and a three-leg diagram can only be the join of at most three primitive divergent overlapping diagrams. For four-leg diagrams in theories with only four-or-less-valent vertices the restriction is even more serve: In these cases, a four-leg diagram can only by the join of at most two primitive overlapping diagrams. \section{Applications to Zero-Dimensional QFT and diagram counting} \label{sec:applications} As an application of the lattice structure, the enumeration of some classes of primitive diagrams using techniques from zero-dimensional quantum field theories is presented. Enumerating diagrams using zero-dimensional QFT is an well-established procedure with a long history \cite{hurst1952enumeration,bender1976statistical,cvitanovic1978number,bessis1980quantum,argyres2001zero} and wide-ranging applications in mathematics \cite{kontsevich1992intersection,lando2013graphs}.  The characteristic property of zero-dimensional QFT is that every diagram in the perturbation expansion has the amplitude $1$. On the Hopf algebra of Feynman diagrams such a prescription can be formulated by the character or \textit{Feynman rule}:  \begin{align} &\phi:& &\hopffg_D& &\rightarrow& &\mathbb{Q}[\hbar],& & \\ && &\Gamma& &\mapsto& &\hbar^{h_1(\Gamma)},& && \end{align} which maps every Feynman diagram to $\hbar$ to the power of its number of loops in the ring of powerseries in $\hbar$. Clearly, $\phi$ is in $G^{{\hopffg_D}}_{\mathbb{Q}[\hbar]}$, the group of characters of $\hopffg_D$ to $\mathbb{Q}[\hbar]$.  Note, that we are not setting $D=0$ even though $\phi$ are the Feynman rules for zero-dimensional QFT. Every diagram would be `convergent' and the Hopf algebra $\hopffg_0$ trivial. It might be more enlightening to think about $\phi$ as toy Feynman rules which assign $1$ to every Feynman diagram without any respect to kinematics. This way we can still study the effects of renormalization on the amplitudes in an arbitrary dimension of spacetime. Moreover, $\phi$ `counts' the number of diagrams weighted by their symmetry factor.

We define the sum of all 1PI diagrams with a certain residue $r$ weighted by their symmetry factor as, \begin{align} X^r:= \sum \limits_{\substack{\Gamma \in \mathcal{T} \\ \res(\Gamma) = r}} \frac{\Gamma}{|\Aut(\Gamma)|}, \end{align} such that $\phi(X^r)$ is the generating function of these weighted diagrams with $\hbar$ as a counting variable. This generating function is the perturbation expansion of the Green's function for the residue $r$. 

The \textit{counter term map} \cite{connes2001renormalization} is defined as, \begin{align} S^R_D := R \circ \phi \circ S_D, \end{align} in a \textit{multiplicative renormalization scheme} $R$ with the antipode $S_D$ of $\hopffg_D$. $S^R_D$ is called the counter term map, because it maps the sum of all 1PI diagrams with a certain residue $r$ to the corresponding counter term, which when substituted into the Lagrangian renormalizes the QFT appropriately. The renormalized Feynman rules are given by the convolution product $\phi^R_D := S^R_D * \phi$. In zero-dimensional QFTs, they vanish on all generators of $\hopffg_D$ except on $\mathbb{I}$. This can be used to obtain differential equations for the $Z$-factors and other interesting quantities as was done in \cite{cvitanovic1978number,argyres2001zero}. 

For the toy Feynman rules $\phi$, there are no kinematics to choose a multiplicative renormalization scheme from. The renormalization will be modeled as usual in the scope of zero-dimensional-QFTs by setting $R=\id$. Consequently, $S^R_D = \phi \circ S_D$. 

The antipode in this formula is the point where the Hopf algebra structure enters the game. The lattice structure can be used to clarify the picture even more and to obtain quantitative results.

We define $\phi' \in G^{\hopflat}_{\mathbb{Q}[\hbar]}$, a Feynman rule on the Hopf algebra of decorated lattices, analogous to $\phi \in G^{{\hopffg_D}}_{\mathbb{Q}[\hbar]}$: \begin{align} &\phi':& &\hopflat& &\rightarrow& &\mathbb{Q}[\hbar],& & \\ && &(P, \nu)& &\mapsto& &\hbar^{\nu(\hat{1})},& && \end{align} which maps a decorated lattice to the value of the decoration of the largest element. Immediately, we can see that $\phi = \phi' \circ \chi_D$. For the counter term map, \begin{align} S^R_D &= \phi' \circ \chi_D \circ S_D, \end{align} is obtained. Using theorem \ref{thm:hopf_alg_morph}, we can commute $\chi_D$ and $S_D$, \begin{align} S^R_D &= \phi' \circ S \circ \chi_D. \end{align} For this reason, the evaluation of $S^R_D$ can be performed entirely in $\hopflat$. $S^R_D$ reduces to a combinatorial calculation on the lattice which is obtained by the Hopf algebra morphism $\chi_D$. The morphism $\phi' \circ S$ maps decorated lattices into the ring of powerseries in $\hbar$. Because $S$ respects the grading in $\nu(\hat{1})$, we can write \begin{align} \phi' \circ S (L, \nu) = \hbar^{\nu(\hat{1})} \zeta \circ S (L, \nu), \end{align} where $\zeta$ is the characteristic function $(L, \nu)\mapsto 1$. The map $\zeta \circ S$ is the \textit{Moebius function}, $\mu(\hat{0},\hat{1})$, on the lattice \cite{ehrenborg1996posets}. It is defined recursively as,  \begin{defn}[Moebius function] \begin{align} \label{eqn:moebius} \mu_P(x,y) &= \begin{cases} 1,&\text{if } x=y \\ - \sum \limits_{x \leq z < y} \mu_P(x,z) & \text{if } x < y. \end{cases} \end{align} for a poset $P$ and $x,y \in P$. \end{defn}

We summarize these observations in  \begin{thm} For zero-dimensional-QFT Feynman rules as $\phi$ above the counter term map takes the form \begin{align} {S^R}'(L,\nu) = \hbar^{\nu(\hat{1})} \mu_L( \hat{0}, \hat{1} ) \end{align} on the Hopf algebra of lattices, where $S^R_D = {S^R}'\circ \chi_D$ and with $\hat{0}$ and $\hat{1}$ the lower and upper bound of $L$. \end{thm} \begin{crll} \begin{align} \label{eqn:phimoebius} {S^R_D}(\Gamma) = \hbar^{h_1(\Gamma)} \mu_{\sdsubdiags_D(\Gamma)}( \hat{0}, \hat{1} ) \end{align} on the Hopf algebra of Feynman diagrams with $\hat{0}=\emptyset$ and $\hat{1}=\Gamma$, the lower and upper bound of $\sdsubdiags_D(\Gamma)$.  \end{crll}

On these grounds, the counter terms in zero-dimensional QFT can be calculated only by computing the Moebius function on the lattice $\sdsubdiags_D(\Gamma)$. The Moebius function is a well studied object in combinatorics. There are especially sophisticated techniques to calculate the Moebius functions on lattices (see \cite{stanley1997,stern1999semimodular}).  For instance \begin{thm}{(Rota's crosscut theorem for atoms and coatoms (special case of \cite[cor. 3.9.4]{stanley1997}))} Let L be a finite lattice and $X$ its set of atoms and $Y$ its set of coatoms, then \begin{align} \mu_L(\hat{0}, \hat{1}) = \sum \limits_k (-1)^k N_k = \sum \limits_k (-1)^k M_k, \end{align} where $N_k$ is the number of $k$-subsets of $X$ whose join is $\hat{1}$ and  $M_k$ is the number of $k$-subsets of $Y$ whose meet is $\hat{0}$. \end{thm} With this theorem the Moebius functions of all the lattices appearing in this article can be calculated very efficiently.

In many cases, an even simpler theorem, which is a special case of the previous one, applies: \begin{thm}{(Hall's theorem \cite[cor. 4.1.7.]{stern1999semimodular})} If in a lattice $\hat{1}$ is not a join of atoms or $\hat{0}$ is not a meet of coatoms, then $\mu(\hat{0}, \hat{1}) = 0$. \end{thm}

In corollary \ref{crll:joinirreducible}, we proved that every vertex-type subdiagram in a QFT with only three-valent vertices is join-irreducible. Hence, it is also not a join of atoms except if it is an atom itself.

\begin{thm} \label{thm:threeQFTmoebvert} In a renormalizable QFT with only three-or-less-valent vertices and $\Gamma$ a vertex-type s.d.\ diagram (i.e. $\Hext(\Gamma)=3$): \begin{align} {S^R_D}(\Gamma)&= \begin{cases} - \hbar^{h_1(\Gamma)}& \text{ if } \Gamma \text{ is primitive} \\ 0& \text{ if } \Gamma \text{ is not primitive.} \end{cases} \end{align} \end{thm} \begin{proof} In both cases the element $\hat{1}=\Gamma$ in the lattice $\sdsubdiags_D(\Gamma)$ is join-irreducible (corollary \ref{crll:joinirreducible}). If $\Gamma$ is primitive $\phi \circ S(\Gamma) = -\phi (\Gamma) = -\hbar^{h_1(\Gamma)}$. If $\Gamma$ is not primitive, it does not cover $\hat{0}$. This implies that $\hat{1}$ is not a join of atoms. Therefore, $\mu_{\sdsubdiags_D(\Gamma)}( \hat{0}, \hat{1} )$ vanishes and so does $S^R_D(\Gamma)$ in accordance to eq. \eqref{eqn:phimoebius}.  \end{proof}

\begin{crll} In a renormalizable QFT with only three-or-less-valent vertices and $r\in \mathcal{R}_v$ a vertex-type residue: \begin{align} {S^R_D}(X^r) = - \phi \circ P_{\text{Prim}(\hopffg_D)}(X^r), \end{align} where $P_{\text{Prim}(\hopffg_D)}$ projects onto the primitive elements of $\hopffg_D$. \end{crll}

Summarizing, we established that in a theory with only three-or-less-valent vertices the counter term ${S^R_D}(X^r)$ counts the number of primitive diagrams if $r\in \mathcal{R}_v$. This fact has been used indirectly in \cite{cvitanovic1978number} to obtain the generating functions for primitive vertex diagrams in $\phi^3$. 

The conventional $Z$-factor for the respective vertex is $Z^r = 1 + r! {S^R_D}(X^r)$, where the factorial of the residue $r=\prod \limits_{\varphi \in \Phi} \varphi^{n_\varphi}$ is $r! = \prod \limits_{\varphi \in \Phi} n_{\varphi}!$. This factorial is necessary to return to the leg-fixed case.

Further exploitation of the lattice structure leads to a statement on propagator-type diagrams in such theories: \begin{thm} In a renormalizable QFT with three-or-less-valent vertices the propagator-type diagrams $\Gamma$, for which $S^R_D(\Gamma) \neq 0$, must have the lattice structure  ${ \def \scale{2ex} \begin{tikzpicture}[x=\scale,y=\scale,baseline={([yshift=-.5ex]current bounding box.center)}] \begin{scope}[node distance=1] \coordinate (top) ; \coordinate [below= of top] (bottom); \draw (top) -- (bottom); \filldraw[fill=white, draw=black] (top) circle(2pt); \filldraw[fill=white, draw=black] (bottom) circle(2pt); \end{scope} \end{tikzpicture} } $ or  ${ \def \scale{2ex} \begin{tikzpicture}[x=\scale,y=\scale,baseline={([yshift=-.5ex]current bounding box.center)}] \begin{scope}[node distance=1] \coordinate (top) ; \coordinate [below left= of top] (v1); \coordinate [below right= of top] (v2); \coordinate [below left= of v2] (bottom); \draw (top) -- (v1); \draw (top) -- (v2); \draw (v1) -- (bottom); \draw (v2) -- (bottom); \filldraw[fill=white, draw=black] (top) circle(2pt); \filldraw[fill=white, draw=black] (v1) circle(2pt); \filldraw[fill=white, draw=black] (v2) circle(2pt); \filldraw[fill=white, draw=black] (bottom) circle(2pt); \end{scope} \end{tikzpicture} } $. \end{thm} \begin{proof} A propagator-type diagram $\Gamma$ either has a maximal forest which is the union of propagator diagrams, has at least two vertex-type subdiagrams or it is the primitive diagram of the topology $\oneloopprop$. In the first case $\Gamma$ is join-irreducible and $S^R_D(\Gamma) = 0$. In the third case the corresponding lattice is ${ \def \scale{2ex} \begin{tikzpicture}[x=\scale,y=\scale,baseline={([yshift=-.5ex]current bounding box.center)}] \begin{scope}[node distance=1] \coordinate (top) ; \coordinate [below= of top] (bottom); \draw (top) -- (bottom); \filldraw[fill=white, draw=black] (top) circle(2pt); \filldraw[fill=white, draw=black] (bottom) circle(2pt); \end{scope} \end{tikzpicture} }$. In the second case, $\Gamma$ covers at least one vertex diagram $\gamma$ which is join-irreducible. Every lattice $L$ with $\mu_L(\hat{0}, \hat{1})\neq0$ is complemented \cite[Cor. 4.1.11]{stern1999semimodular}. In a complemented lattice $L$, there is a $y \in L$ for every $x \in L$ such that $x \join y = \hat{1}$ and $y \meet x = \hat{0}$. For this reason, all the join-irreducible elements of $L$ must be atoms if $\mu_L(\hat{0}, \hat{1})\neq0$. As was shown in the proof of proposition \ref{prop:semimodular}, a propagator cannot be the join of more than two primitive diagrams. Accordingly, ${ \def \scale{2ex} \begin{tikzpicture}[x=\scale,y=\scale,baseline={([yshift=-.5ex]current bounding box.center)}] \begin{scope}[node distance=1] \coordinate (top) ; \coordinate [below left= of top] (v1); \coordinate [below right= of top] (v2); \coordinate [below left= of v2] (bottom); \draw (top) -- (v1); \draw (top) -- (v2); \draw (v1) -- (bottom); \draw (v2) -- (bottom); \filldraw[fill=white, draw=black] (top) circle(2pt); \filldraw[fill=white, draw=black] (v1) circle(2pt); \filldraw[fill=white, draw=black] (v2) circle(2pt); \filldraw[fill=white, draw=black] (bottom) circle(2pt); \end{scope} \end{tikzpicture} } $ is the only possible lattice if $\Gamma$ is not primitive. \end{proof}

The $Z$-factors for the propagators can also be obtained using the last theorem.  To do this, the Moebius function for each propagator diagram must be calculated using the form of the lattices and eq. \eqref{eqn:moebius}. The Moebius functions for the vertex-type diagrams are known from theorem \ref{thm:threeQFTmoebvert}.  \begin{expl} In a renormalizable QFT with only three-or-less-valent vertices and $r\in \mathcal{R}_e$, a propagator-type residue: \begin{align} \begin{split} {S^R_D}(X^r) &= \hbar \sum \limits_{\substack{\res \Gamma_P = r \\ \left\{v_1,v_2\right\} = V(\Gamma_P)}} \frac{1}{\Aut \Gamma_P} \left( -1 + \frac{ \res(v_1)!}{2} \phi \circ P_{\text{Prim}(\hopffg_D)}(X^{\res(v_1)}) \right. + \\ &+ \left. \frac{\res(v_2)!}{2}\phi \circ P_{\text{Prim}(\hopffg_D)}(X^{\res(v_2)}) \right) \end{split} \end{align} where the sum is over all primitive propagator diagrams $\Gamma_P$ with a topology as $\oneloopprop$ and exactly two vertices $v_1,v_2$. Of course, this sum is finite.

The factorials of the residues must be included to fix the external legs of the vertex-type subdiagrams. The factor of $\frac12$ is necessary, because every non-primitive diagram, which contributes to the counter term, has exactly two maximal forests. This is an example of a simple Dyson-Schwinger equation in the style of \cite{Kreimer2006}.

The conventional $Z$-factor for the propagator is $Z^r = 1 + r! {S^R_D}(X^r)$. \end{expl}

Although the counter term map for renormalizable QFTs with only three-or-less-valent vertices enumerates primitive diagrams, we cannot assume that the situation is similar in a more general setting with also four-valent vertices. A negative result in this direction was obtained in \cite[p. 27]{argyres2001zero} by Argyres et al. They observed that the vertex counter term in zero-dimensional $\phi^4$ does not count primitive diagrams. 

\begin{figure} \ifdefined\nodraft \subcaptionbox{Structure of overlapping divergences of four-leg diagrams in theories with four-or-less-valent vertices.\label{fig:fourvalblob}}[.45\linewidth]{ \def \scale {5ex} \begin{tikzpicture}[x=\scale,y=\scale,baseline={([yshift=-.5ex]current bounding box.center)}] \begin{scope}[node distance=1] \coordinate (i0); \coordinate[below=1 of i0] (i1); \coordinate[right= of i0] (vlt); \coordinate[right= of i1] (vlb); \coordinate[right=2 of vlt] (vmt); \coordinate[right=2 of vlb] (vmb); \coordinate[right=2 of vmt] (vrt); \coordinate[right=2 of vmb] (vrb); \coordinate[right= of vrt] (o0); \coordinate[right= of vrb] (o1); \draw (i0) -- (vlt); \draw (i1) -- (vlb); \draw (vmt) -- (vlt); \draw (vmb) -- (vlb); \draw (vmt) -- (vrt); \draw (vmb) -- (vrb); \draw (o0) -- (vlt); \draw (o1) -- (vlb); \draw [dashed] ($(vlt) + (1,.5)$) -- ($(vlb) + (1,-.5)$); \draw [dashed] ($(vmt) + (1,.5)$) -- ($(vmb) + (1,-.5)$); \draw [fill=white] ($(vlt) + (.5,.5)$) rectangle ($(vlb) - (.5,.5)$) ; \draw [fill=white,thick,pattern=north west lines] ($(vlt) + (.5,.5)$) rectangle ($(vlb) - (.5,.5)$) ; \draw [fill=white] ($(vmt) + (.5,.5)$) rectangle ($(vmb) - (.5,.5)$) ; \draw [fill=white,thick,pattern=north east lines] ($(vmt) + (.5,.5)$) rectangle ($(vmb) - (.5,.5)$) ; \draw [fill=white,thick,pattern=north west lines] ($(vmt) + (.5,.5)$) rectangle ($(vmb) - (.5,.5)$) ; \draw [fill=white] ($(vrt) + (.5,.5)$) rectangle ($(vrb) - (.5,.5)$) ; \draw [fill=white,thick,pattern=north west lines] ($(vrt) + (.5,.5)$) rectangle ($(vrb) - (.5,.5)$) ; \end{scope} \end{tikzpicture} }
\subcaptionbox{Chain of overlapping divergences of the type in fig. \ref{fig:fourvalblob}, which will evaluate to a non-zero Moebius function.\label{fig:fourvalblobchain}}[.45\linewidth]{ \def \scale {5ex} \begin{tikzpicture}[x=\scale,y=\scale,baseline={([yshift=-.5ex]current bounding box.center)}] \begin{scope}[node distance=1] \coordinate (i0); \coordinate[below=1 of i0] (i1); \coordinate[right= of i0] (vlt); \coordinate[right= of i1] (vlb); \coordinate[right=2 of vlt] (vmt); \coordinate[right=2 of vlb] (vmb); \coordinate[right=2 of vmt] (vmmt); \node[below=.5 of vmmt] (vmm) {$\ldots$}; \coordinate[right=2 of vmb] (vmmb); \coordinate[right=2 of vmmt] (vrt); \coordinate[right=2 of vmmb] (vrb); \coordinate[right= of vrt] (o0); \coordinate[right= of vrb] (o1); \draw (i0) -- (vlt); \draw (i1) -- (vlb); \draw (vmt) -- (vlt); \draw (vmb) -- (vlb); \draw[dotted] (vmt) -- (vrt); \draw[dotted] (vmb) -- (vrb); \draw (vmt) -- ($(vmmt) - (1,0)$); \draw (vmb) -- ($(vmmb) - (1,0)$); \draw (vrt) -- ($(vmmt) + (1,0)$); \draw (vrb) -- ($(vmmb) + (1,0)$); \draw (o0) -- (vrt); \draw (o1) -- (vrb); \draw [fill=white] ($(vlt) + (.5,.5)$) rectangle ($(vlb) - (.5,.5)$) ; \draw [fill=white,thick,pattern=north west lines] ($(vlt) + (.5,.5)$) rectangle ($(vlb) - (.5,.5)$) ; \draw [fill=white] ($(vmt) + (.5,.5)$) rectangle ($(vmb) - (.5,.5)$) ; \draw [fill=white,thick,pattern=north west lines] ($(vmt) + (.5,.5)$) rectangle ($(vmb) - (.5,.5)$) ; \draw [fill=white] ($(vrt) + (.5,.5)$) rectangle ($(vrb) - (.5,.5)$) ; \draw [fill=white,thick,pattern=north west lines] ($(vrt) + (.5,.5)$) rectangle ($(vrb) - (.5,.5)$) ; \end{scope} \end{tikzpicture} }
\else

MISSING IN DRAFT MODE

\fi \caption{Overlapping divergences for diagrams with four legs in theories with only four-or-less-valent vertices.} \end{figure}
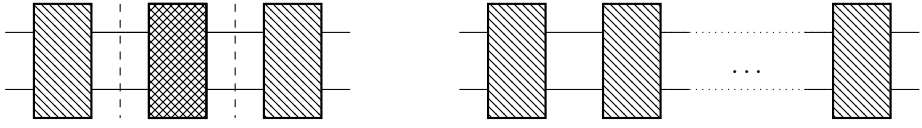

From the perspective of lattice theory this result can be explained. The only way in which overlapping divergences can appear in a diagram with four legs in a QFT with only four-or-less-valent vertices is depicted in fig. \ref{fig:fourvalblob}. The dashed lines indicate the possible cuts to separate one overlapping divergence from the other. To obtain a Feynman diagram $\Gamma$ with $\mu_{\sdsubdiags_D(\Gamma)}(\hat{0},\hat{1})\neq 0$, the blob in the middle must be either of the same overlapping type as fig. \ref{fig:fourvalblob} or superficially convergent. Otherwise, a join-irreducible element would be generated, which would imply $\mu_{\sdsubdiags_D(\Gamma)}(\hat{0},\hat{1}) = 0$.  The possible non-primitive diagrams with four legs, which give a non-vanishing Moebius function are consequently of the form depicted in fig. \ref{fig:fourvalblobchain}, where each blob must be replaced by a superficially convergent four-leg diagram such that the diagram remains 1PI. The lattice corresponding to this structure of overlapping divergences is a boolean lattice. Every superficially divergent subdiagram can be characterized by the particular set partition of `blocks' it contains. This gives a bijection from $\{0,1\}^n$ to all possible subdivergences. The Moebius function of boolean lattices evaluates to $(-1)^n$, where $n$ is the number of atoms \cite[Ex. 3.8.4.]{stanley1997}.  Accordingly, the structure of overlapping four-leg diagrams depends on the number of superficially convergent four-leg diagrams. The situation is especially simple in $\phi^4$-theory: \begin{expl}[Overlapping vertex-type diagrams in $\phi^4$-theory] The only superficially convergent four-leg diagram in $\phi^4$-theory is the single four-leg vertex $\fourvtx$ and so the only vertex-type s.d.\ diagrams $\Gamma$ which give $\mu_{\sdsubdiags_D(\Gamma)}(\hat{0},\hat{1})\neq 0$ are  \ifdefined\nodraft ${ \def \scale {1ex} \begin{tikzpicture}[x=2ex,y=2ex,baseline={([yshift=-.5ex]current bounding box.center)}] \coordinate (v0); \coordinate [right=.5 of v0] (vm1); \coordinate [right=.5 of vm1] (v1); \node [right=.5 of v1] (v2) {$\ldots$}; \coordinate [right=.5 of v1] (v1m); \coordinate [right=.01 of v2] (v2m); \coordinate [right=.5 of v2m] (v3); \coordinate [right=.5 of v3] (vm2); \coordinate [right=.5 of vm2] (v4); \coordinate [above left=.5 of v0] (i0); \coordinate [below left=.5 of v0] (i1); \coordinate [above right=.5 of v4] (o0); \coordinate [below right=.5 of v4] (o1); \draw (vm1) circle(.5); \draw (vm2) circle(.5); \draw (i0) -- (v0); \draw (i1) -- (v0); \draw (o0) -- (v4); \draw (o1) -- (v4); \draw ([shift=(90:.5)]v1m) arc (90:270:.5); \draw ([shift=(-90:.5)]v2m) arc (-90:90:.5); \filldraw (v0) circle(1pt); \filldraw (v1) circle(1pt); \filldraw (v3) circle(1pt); \filldraw (v4) circle(1pt); \end{tikzpicture} }$ \else

MISSING IN DRAFT MODE

\fi , the chains of one-loop diagrams. Their generating function is $\frac18 \sum \limits_{L\geq 0} \frac{\hbar^L }{2^{L}}$ (every diagram is weighted by its symmetry factor) and the counter term map in $\phi^4$ evaluates to, \begin{align} \label{eqn:phird} {S^R_D}(X^\fourvtx) = - \phi \circ P_{\text{Prim}(\hopffg)}(X^{\fourvtx}) + \frac18 \sum \limits_{L\geq 2} (-1)^L \left( \frac{\hbar }{2} \right)^L. \end{align} Note again, that in this setup the legs of the diagrams are not fixed as explained in section \ref{sec:feynmandiagrams}. To reobtain the numbers for the case with fixed legs, the generating function must be multiplied with the typical value $4!=24$.  The formula for the usual vertex $Z$-factor is $Z^\fourvtx = 1+ 4! {S^R_D}(X^\fourvtx)$.

Using this result, we can indeed use the counter terms of zero-dimensional $\phi^4$-theory to calculate the number of primitive diagrams. We merely must include the correction term on the right-hand side of eq. \eqref{eqn:phird}. \end{expl}

\begin{expl}[Overlapping four-leg-vertex diagrams in pure Yang-Mills theory] In pure Yang-Mills theory there can be either the single four-valent vertex $\fourvtxgluon$ or two three-valent vertices joined by a propagator $\twothreevtxgluon$ as superficially convergent four-leg diagrams. Only chains of diagrams as in fig. \ref{fig:fourvalblobchain} or primitive diagrams give a non-zero ${S^R_D}(\Gamma)$. At two loop for instance, the non-primitive diagrams  \ifdefined\nodraft \begin{align*} { \def \scale{3ex} 
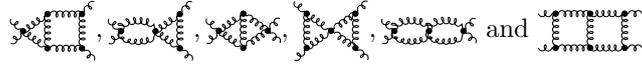
\begin{tikzpicture}[x=\scale,y=\scale,baseline={([yshift=-.5ex]current bounding box.center)}] \begin{scope}[node distance=1] \coordinate (v0); \coordinate [right=.707 of v0] (vm); \coordinate [above=.5 of vm] (vt1); \coordinate [below=.5 of vm] (vb1); \coordinate [right=1 of vt1] (vt2); \coordinate [right=1 of vb1] (vb2); \coordinate [above left=.5 of v0] (i0); \coordinate [below left=.5 of v0] (i1); \coordinate [above right=.5 of vt2] (o0); \coordinate [below right=.5 of vb2] (o1); \draw[gluon] (v0) -- (vt1); \draw[gluon] (v0) -- (vb1); \draw[gluon] (vt1) -- (vb1); \draw[gluon] (vt2) -- (vb2); \draw[gluon] (vt2) -- (vt1); \draw[gluon] (vb2) -- (vb1); \draw[gluon] (i0) -- (v0); \draw[gluon] (i1) -- (v0); \draw[gluon] (o0) -- (vt2); \draw[gluon] (o1) -- (vb2); \filldraw (v0) circle(1pt); \filldraw (vt1) circle(1pt); \filldraw (vb1) circle(1pt); \filldraw (vt2) circle(1pt); \filldraw (vb2) circle(1pt); \end{scope} \end{tikzpicture} }, { \def \scale{3ex} \begin{tikzpicture}[x=\scale,y=\scale,baseline={([yshift=-.5ex]current bounding box.center)}] \begin{scope}[node distance=1] \coordinate (v0); \coordinate [right=.5 of v0] (vm); \coordinate [right=.5 of vm] (v1); \coordinate [above left=.5 of v0] (i0); \coordinate [below left=.5 of v0] (i1); \coordinate [right=.707 of v1] (vm2); \coordinate [above=.5 of vm2] (vt); \coordinate [below=.5 of vm2] (vb); \coordinate [above right=.5 of vt] (o0); \coordinate [below right=.5 of vb] (o1); \draw[gluon] (v1) -- (vt); \draw[gluon] (v1) -- (vb); \draw[gluon] (vt) -- (vb); \draw[gluon] (v0) to[bend left=90] (v1); \draw[gluon] (v0) to[bend right=90] (v1); \draw[gluon] (i0) -- (v0); \draw[gluon] (i1) -- (v0); \draw[gluon] (o0) -- (vt); \draw[gluon] (o1) -- (vb); \filldraw (v0) circle(1pt); \filldraw (v1) circle(1pt); \filldraw (vt) circle(1pt); \filldraw (vb) circle(1pt); \end{scope} \end{tikzpicture} }, { \def \scale{3ex} \begin{tikzpicture}[x=\scale,y=\scale,baseline={([yshift=-.5ex]current bounding box.center)}] \begin{scope}[node distance=1] \coordinate (v0); \coordinate [right=.707 of v0] (vm); \coordinate [above=.5 of vm] (vt); \coordinate [below=.5 of vm] (vb); \coordinate [right=.707 of vm] (v1); \coordinate [above left=.5 of v0] (i0); \coordinate [below left=.5 of v0] (i1); \coordinate [above right=.5 of v1] (o0); \coordinate [below right=.5 of v1] (o1); \draw[gluon] (v0) -- (vt); \draw[gluon] (v0) -- (vb); \draw[gluon] (vt) -- (vb); \draw[gluon] (v1) -- (vt); \draw[gluon] (v1) -- (vb); \draw[gluon] (i0) -- (v0); \draw[gluon] (i1) -- (v0); \draw[gluon] (o0) -- (v1); \draw[gluon] (o1) -- (v1); \filldraw (v0) circle(1pt); \filldraw (v1) circle(1pt); \filldraw (vt) circle(1pt); \filldraw (vb) circle(1pt); \end{scope} \end{tikzpicture} }, { \def \scale{3ex} \begin{tikzpicture}[x=\scale,y=\scale,baseline={([yshift=-.5ex]current bounding box.center)}] \begin{scope}[node distance=1] \coordinate (v0); \coordinate [above=.5 of v0] (vt1); \coordinate [below=.5 of v0] (vb1); \coordinate [right=.707 of v0] (vm); \coordinate [right=.707 of vm] (v1); \coordinate [above=.5 of v1] (vt2); \coordinate [below=.5 of v1] (vb2); \coordinate [above left=.5 of vt1] (i0); \coordinate [below left=.5 of vb1] (i1); \coordinate [above right=.5 of vt2] (o0); \coordinate [below right=.5 of vb2] (o1); \draw[gluon] (vt1) -- (vm); \draw[gluon] (vb1) -- (vm); \draw[gluon] (vt2) -- (vm); \draw[gluon] (vb2) -- (vm); \draw[gluon] (vt1) -- (vb1); \draw[gluon] (vt2) -- (vb2); \draw[gluon] (i0) -- (vt1); \draw[gluon] (i1) -- (vb1); \draw[gluon] (o0) -- (vt2); \draw[gluon] (o1) -- (vb2); \filldraw (vt1) circle(1pt); \filldraw (vt2) circle(1pt); \filldraw (vm) circle(1pt); \filldraw (vb1) circle(1pt); \filldraw (vb2) circle(1pt); \end{scope} \end{tikzpicture} }, { \def \scale{3ex} \begin{tikzpicture}[x=\scale,y=\scale,baseline={([yshift=-.5ex]current bounding box.center)}] \begin{scope}[node distance=1] \coordinate (v0); \coordinate [right=.5 of v0] (vm1); \coordinate [right=.5 of vm1] (v1); \coordinate [right=.5 of v1] (vm2); \coordinate [right=.5 of vm2] (v2); \coordinate [above left=.5 of v0] (i0); \coordinate [below left=.5 of v0] (i1); \coordinate [above right=.5 of v2] (o0); \coordinate [below right=.5 of v2] (o1); \draw[gluon] (v0) to[bend left=90] (v1); \draw[gluon] (v0) to[bend right=90] (v1); \draw[gluon] (v1) to[bend left=90] (v2); \draw[gluon] (v1) to[bend right=90] (v2); \draw[gluon] (i0) -- (v0); \draw[gluon] (i1) -- (v0); \draw[gluon] (o0) -- (v2); \draw[gluon] (o1) -- (v2); \filldraw (v0) circle(1pt); \filldraw (v1) circle(1pt); \filldraw (v2) circle(1pt); \end{scope} \end{tikzpicture} } \text{ and } { \def \scale{3ex} \begin{tikzpicture}[x=\scale,y=\scale,baseline={([yshift=-.5ex]current bounding box.center)}] \begin{scope}[node distance=1] \coordinate (i0); \coordinate[below=1 of i0] (i1); \coordinate[right=.5 of i0] (vlt); \coordinate[right=.5 of i1] (vlb); \coordinate[right=1 of vlt] (vmt); \coordinate[right=1 of vlb] (vmb); \coordinate[right=1 of vmt] (vrt); \coordinate[right=1 of vmb] (vrb); \coordinate[right=.5 of vrt] (o0); \coordinate[right=.5 of vrb] (o1); \draw[gluon] (i0) -- (vlt); \draw[gluon] (i1) -- (vlb); \draw[gluon] (vmt) -- (vlt); \draw[gluon] (vmb) -- (vlb); \draw[gluon] (vmt) -- (vrt); \draw[gluon] (vmb) -- (vrb); \draw[gluon] (vlt) -- (vlb); \draw[gluon] (vmt) -- (vmb); \draw[gluon] (vrt) -- (vrb); \draw[gluon] (o0) -- (vrt); \draw[gluon] (o1) -- (vrb); \filldraw (vlt) circle(1pt); \filldraw (vlb) circle(1pt); \filldraw (vmt) circle(1pt); \filldraw (vmb) circle(1pt); \filldraw (vrt) circle(1pt); \filldraw (vrb) circle(1pt); \end{scope} \end{tikzpicture} } \end{align*} \else MISSING IN DRAFT MODE \fi contribute non-trivially to ${S^R_D}(X^\fourvtxgluon)$. These are the only four-leg diagrams which can be formed as the union of two primitive diagrams in this theory. The generating function for $L \geq 2$ of these diagrams is $\frac38 \sum \limits_{L\geq 2} \left(\frac{3 \hbar }{2}\right)^{L}$. Hence, the counter term map in pure Yang-Mills theory for the four-gluon amplitude in zero-dimensional QFT evaluates to, \begin{align} {S^R_D}(X^\fourvtxgluon) = - \phi \circ P_{\text{Prim}(\hopffg)}(X^{\fourvtxgluon}) + \frac38 \sum \limits_{L\geq 2} (-1)^L \left( \frac{3\hbar}{2} \right)^{L}. \end{align} To reobtain the numbers for the case with fixed legs this generating function needs to by multiplied with the value $4!=24$ as in the example for $\phi^4$-theory.  The formula for the $Z$-factor is $Z^\fourvtxgluon = 1 + 4! {S^R_D}(X^\fourvtxgluon)$. \end{expl} \section{Outlook}

It was shown that the Hopf algebra morphism $\chi_D$ can be used to disentangle the combinatorial and the analytic part of Feynman diagram calculations. The Hopf algebra of Feynman diagrams is graded by the coradical degree of its elements as demonstrated in theorem \ref{thm:gradthree} for the case of three-or-less-valent theories and theorem \ref{thm:gradfour} for four-or-less-valent theories with the tadpole diagrams set to zero. 

This grading could also be used to reorganize diagrams which `renormalize in the same way'. The preimage $\chi_D^{-1}$ of some a decorated lattice in $\hopflats$ is a space of such diagrams. With methods from \cite{brown2013angles} this could be used to the express the log-expansion of Green functions systematically. Primitive diagrams of coradical degree one contribute to the first power in the log-expansion, diagrams of coradical degree two to the second and diagrams with a coradical degree equal to the loop number contribute to the leading-log \cite{kruger2015filtrations}. 

Furthermore, this framework can be used to make more statements and explicit calculations on the weighted numbers of primitive diagrams in different QFTs and their asymptotic behavior. These aspects will be analyzed from a combinatorial perspective in a subsequent article \cite{borinskyunpub}. 

\section*{Acknowledgements} I thank Dirk Kreimer for his great supervision, full support and encouragement; David Broadhurst for sparking my interest in the subject of zero-dimensional QFT, for many helpful conversations and motivation to write these results down; Erik Panzer for discussions and for finding the diagram in fig. \ref{fig:nolattice}, with a set of subdivergences which does not form a lattice; as well as Dominique Manchon and Marko Berghoff for fruitful discussions on the subject.

\bibliography{literature}

\end{document}